\documentclass[journal,draftcls,onecolumn,11pt]{IEEEtran}
\usepackage{subfigure,cite,graphicx,psfig,amsmath,amssymb,eufrak,mathrsfs,epsf,epsfig}

\usepackage[active]{srcltx}
\ifCLASSINFOpdf
  
\else
\fi

\begin{document}
\newtheorem{ach}{Achievability}
\newtheorem{con}{Converse}
\newtheorem{definition}{Definition}
\newtheorem{theorem}{Theorem}
\newtheorem{lemma}{Lemma}
\newtheorem{example}{Example}
\newtheorem{cor}{Corollary}
\newtheorem{prop}{Proposition}
\newtheorem{conjecture}{Conjecture}
\newtheorem{remark}{Remark}
\title{Degrees of Freedom of the Rank-deficient Interference Channel with Feedback}
\author{\IEEEauthorblockN{Sung Ho Chae,~\IEEEmembership{Student Member,~IEEE}, Changho Suh,~\IEEEmembership{Member,~IEEE}, and Sae-Young
Chung,~\IEEEmembership{Senior Member,~IEEE}}\thanks{The material
in this paper will be presented in part at the IEEE International
Symposium on Information Theory (ISIT)
 2013 and was in part submitted to the Allerton Conference on Communication, Control, and Computing 2013.}%
}
 \maketitle

\begin{abstract}
We investigate the total degrees of freedom (DoF) of the $K$-user
rank-deficient interference channel with feedback. For the two-user
case, we characterize the total DoF by developing an achievable
scheme and deriving a matching upper bound. For the three-user case,
we develop a new achievable scheme which employs interference alignment
to  efficiently utilize the dimension of the received signal space. In addition, we derive an
upper bound for the general $K$-user case and show the tightness of the bound when the number of
antennas at each node is sufficiently large. As a
consequence of these results, we show that feedback can increase the
DoF when the number of antennas at each node is large enough as
compared to the ranks of channel matrices. This finding is in
contrast to the full-rank interference channel where feedback
provides no DoF gain. The gain comes from using feedback to provide
alternative signal paths, thereby effectively increasing the ranks
of desired channel matrices.

\end{abstract}
 \IEEEpeerreviewmaketitle

\section{Introduction}
It is well known that feedback cannot increase the capacity of
memoryless point-to-point channels~\cite{Shannon56}. Although the
capacity of multiple access channels can in fact increase when
feedback is present, the gain is bounded by one bit for the
Gaussian case~\cite{Ozarow84}. These results give a pessimistic
view on feedback capacity, although feedback can still be useful
for simplifying coding strategies as well as improving
reliability~\cite{LNIT}. Recent work~\cite{Suh11}, however, has
shown that in interference channels, feedback can provide more
significant gains. Specifically, it is shown that the capacity
gain due to feedback becomes arbitrarily large for certain channel
parameters (unbounded gain). The gain comes from the fact that
feedback can help efficient resource sharing between the
interfering users. In the process of deriving this conclusion,
\cite{Suh11} has characterized the feedback capacity region to
within 2 bits of the two-user Gaussian interference channel.

The results of~\cite{Suh11} indicate that feedback enables a
significant capacity improvement of multi-user networks with
interfering links. However, if we turn our attention to degrees of
freedom (DoF), feedback fails to provide promising results. From the
results of~\cite{Huang09,Vaze11}, it has been shown that feedback
cannot improve the total DoF for the two-user full-rank Gaussian
MIMO interference channel\footnote{However, recently it has been
shown in~\cite{Hwang12} that for \emph{multihop} networks, feedback
can increase DoF.}. Therefore, feedback can provide unbounded
capacity gain but cannot increase the DoF in the full-rank channel.

In this work, we show that feedback, however, can increase the total
DoF in the \emph{rank-deficient} interference channel. The
rank-deficient channel captures a poor scattering environment where
there are only few signal paths between nodes. For
example, for rooftop-to-rooftop communications in which transmit and
receive antennas are mounted high above the ground, the angular
spread becomes very low~\cite{Winner,802.16m,mit,rooftop}, and, as a
result, the channel matrix becomes rank deficient due to the lack of
multipath. In addition, for massive MIMO communications in which
each node has numerous antennas, channel matrices cannot have full
rank unless there are enough number of signal paths between nodes.
The non-feedback DoF of the rank-deficient interference channel has been
studied in~\cite{Chae11,Krishnamurthy12,Zeng12}, and the optimal
DoFs for the two-user and three-user cases have been established
in~\cite{Krishnamurthy12}. In this paper, we now investigate the effects of \emph{feedback} on the total DoF of the rank-deficient interference channel. 
For the two-user case, we adopt the same rank-deficient channel
model as in~\cite{Krishnamurthy12} in which the number of transmit
and receive antennas and the ranks of channel matrices are
arbitrary. We develop an achievable scheme and also derive a
matching upper bound, thus characterizing the total DoF. For the
three-user case, we focus on a symmetric case in which each node
has the same number of antennas, the rank of each direct link is
the same, and the rank of each cross link is the same. We
establish the achievable total DoF of this channel by developing a
new achievable scheme. The proposed scheme employs interference
alignment to efficiently utilize the dimension of the received signal space
when the rank of cross links is sufficiently large as compared to
the number of antennas at each node. Furthermore, we derive an
upper bound for the general $K$-user case, which is indeed achievable 
when the number of antennas at each node is sufficiently large. As
a consequence of these results, we show that feedback can increase
the DoF when channel matrices of desired links are highly
rank-deficient. The gain comes from the fact that feedback can
provide alternative signal paths in the rank-deficient channel,
and hence the ranks of desired channel matrices are effectively
increased, which is not possible in the full-rank channel. The
result of this paper also includes that of the full-rank channel
with feedback as a special case.

The rest of this paper is organized as follows. In Section II, we
describe the channel model considered in the paper. In Section III,
we show the main results of the paper and provide an intuition as to how
feedback can increase the DoF in the rank-deficient channel. In addition,
we provide the proofs of main theorems in Sections IV, V, and VI.
Finally, we conclude the paper in Section VII.

\textbf{Notations}: Throughout the paper, we will use $\mathbf{A}$
and $\mathbf{a}$ to denote a matrix and a vector, respectively.
Let $\mathbf{A}^{T}$ and $||\mathbf{A}||$ denote the transpose and
the norm of $\mathbf{A}$, respectively. In addition, let
$|\mathbf{A}|$ and $\text{rank}(\mathbf{A})$ denote the
determinant and the rank of $\mathbf{A}$, respectively. The
notations $\mathbf{I}_n$ and $\mathbf{0}_{n\times n}$ denote
the $n\times n$ identity matrix and zero matrix, respectively.
We write $f(x)=o(x)$ if $\lim_{x\rightarrow
\infty}\frac{f(x)}{x}=0$. For convenience, when indexing channel matrices, we use modular arithmetic where the modulus is the number of users. (e.g., for the three-user case, $\mathbf{H}_{1,4}$ means $\mathbf{H}_{1,1}$).

\section{System Model}
Consider the $K$-user rank-deficient interference channel with feedback, as
depicted in Fig. \ref{fig1}. Transmitter $i$ wishes to
communicate with receiver $i$, and
transmitter $i$ and receiver $i$ use $M_i$ and $N_i$ antennas, respectively. We assume that all channel
coefficients are fixed and known to all nodes. Then, the input and
output relationship at time slot $t$ is given by
\begin{equation*}
\mathbf{y}_{j}(t)=\sum^{K}_{i=1}\mathbf{H}_{j,i}\mathbf{x}_{i}(t)+\mathbf{z}_{j}(t),
\end{equation*}
where $\mathbf{x}_{i}(t)$ is the $M_i\times 1$ input signal vector
at transmitter $i$, $\mathbf{H}_{j,i}$ is the $N_j\times M_i$
channel matrix from transmitter $i$ to receiver $j$, and
$\mathbf{y}_j(t)$ is the $N_j\times 1$ received signal vector at
receiver $j$. The noise vector $\mathbf{z}_j(t)$ is the additive
white circularly symmetric complex Gaussian with zero mean and
covariance of $\mathbf{I}_{N_j}$. We assume that all of the noise
vectors and signal vectors are independent of each other.

In this paper, we adopt the rank-deficient channel model
in~\cite{Tse_wireless}, in which there are $D_{j,i}\leq
\min\{M_i,N_j\}$ independent signal paths from transmitter $i$ to
receiver $j$. Let $\mathbf{H}^{(k)}_{j,i}$ denote the channel matrix
corresponding to the $k$th signal path between transmitter $i$ and receiver $j$. Note that due to the
key-hole effect~\cite{Tse_wireless},
$\text{rank}(\mathbf{H}^{(k)}_{j,i})=1$, $\forall
k=1,2,\ldots,D_{j,i}$. Therefore, we assume that the matrix
$\mathbf{H}_{j,i}$ is given by
\begin{align}
\mathbf{H}_{j,i}&=\sum_{k=1}^{D_{j,i}}\mathbf{H}^{(k)}_{j,i}\nonumber \\
&=\sum_{k=1}^{D_{j,i}}\mathbf{a}^{(k)}_{j,i}{\mathbf{b}_{j,i}^{(k)}}^T, \quad
\forall i,j=1,2,\ldots,K
\end{align}
where $\mathbf{a}_{j,i}^{(k)}$ and $\mathbf{b}^{(k)}_{j,i}$ are $N_j\times
1$ and $M_i \times 1$ vectors respectively, and their coefficients
are drawn from a continuous distribution. From (1), we can see
that $\text{rank}(\mathbf{H}_{j,i})=D_{j,i}$ with probability one.

\begin{figure}[!t]
\centering
\includegraphics[scale=0.7]{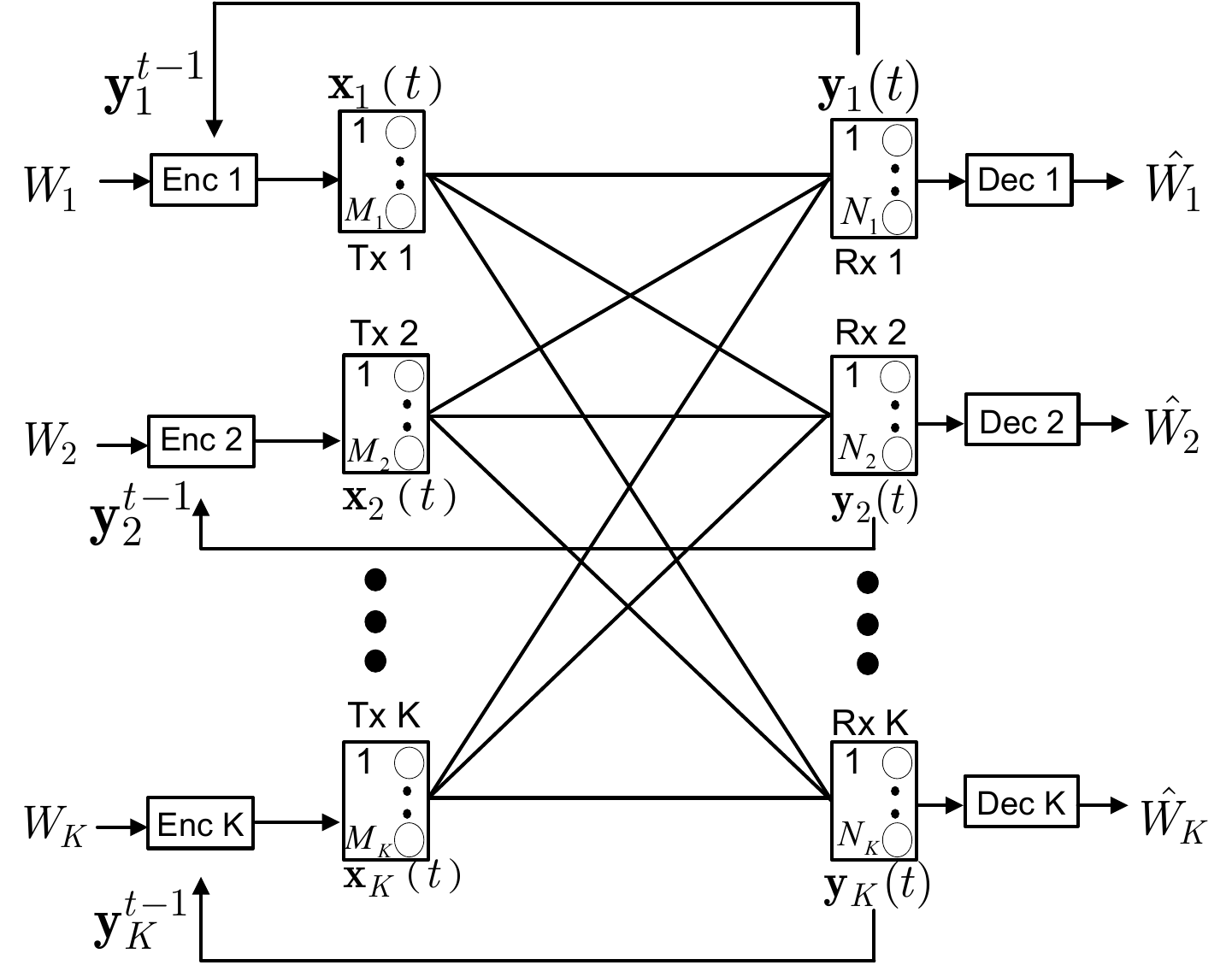}
\caption{The $K$-user rank-deficient interference channel with feedback.}
\label{fig1} \vspace{-0.1in}
\end{figure}

There are $K$ independent messages $W_1,W_2,\ldots, W_K$. At time slot
$t$, transmitter $i$ sends the encoded signal $\mathbf{x}_{i}(t)$,
which is a function of $W_i$ and past output sequences
$\mathbf{y}_i^{t-1}\triangleq [\begin{array}{ccccc}
 \mathbf{y}_i(1)& \mathbf{y}_i(2)& \cdots &\mathbf{y}_i(t-1) \\
\end{array}]^{T}$. We assume that each transmitter should satisfy the average power constraint $P$, i.e.,
$E[|\mathbf{x}_i(t)|^2]\leq P$ for $i\in\{1,2,\ldots, K\}$.  A rate tuple
$(R_1, R_2, \ldots, R_K)$ is said to be achievable if there exists a sequence
of $(2^{nR_1},2^{nR_2}, \ldots ,,2^{nR_K}, n)$ codes such that the average probability
of decoding error tends to zero as the code length $n$ goes to
infinity. The capacity region $\mathcal{C}$ of this channel is the
closure of the set of achievable rate tuples $(R_1, R_2, \ldots, R_K)$. The
total DoF is defined as $\Gamma=\lim_{P \to
\infty}\max_{(R_1,R_2,\ldots, R_K)\in \mathcal{C}}\frac{\sum_{i=1}^{K}R_i}{\log(P)}$.

\section{Main Results}
\subsection{Two-user case}
For the two-user case, we characterize the total DoF as stated in the following theorem by developing an achievable scheme and deriving a
matching upper bound.
\begin{theorem}[Two-user case]
For the two-user rank-deficient interference channel with
feedback, the total DoF is given by
\begin{align*}
\Gamma_{fb}=&\min\{M_1+N_2-D_{2,1},M_2+N_1-D_{1,2},\\
&D_{1,1}+D_{2,2}+D_{1,2}, D_{1,1}+D_{2,2}+D_{2,1},\\
&\min\{M_1,N_1\}+D_{2,2},\min\{M_2,N_2\}+D_{1,1}\}
\end{align*}
\begin{proof}
See Section IV for the proof.
\end{proof}
\end{theorem}
\vspace{0.05in}
\begin{remark}[Full-rank case]
 For the case in which all the channel matrices have
full ranks, i.e., $D_{j,i}=\min(M_i,N_j)$ $\forall i,j=1,2$, the
total DoF becomes
\begin{align*}
\Gamma_{fb}=&\min\{M_1+M_2,N_1+N_2,\max\{M_1,N_2\},\max\{M_2,N_1\}\},
\end{align*}
which coincides with the result for the full-rank interference
channel in~\cite{Huang09,Vaze11}.
\end{remark}

\vspace{0.05in}
\vspace{0.05in}
\begin{remark}
If all the direct links have full ranks, i.e.,
$D_{1,1}=\min(M_1,N_1)$ and $D_{2,2}=\min(M_2,N_2)$, the result
recovers the non-feedback case in~\cite{Krishnamurthy12}:
\begin{align*}
\Gamma_{no}=&\min\{M_1+N_2-D_{2,1},N_1+M_2-D_{1,2},D_{1,1}+D_{2,2}\}.
\end{align*}
\end{remark}
\vspace{0.05in}
Notice that for the above two cases, feedback cannot increase the
total DoF.
\vspace{0.05in}

\textbf{DoF gain due to feedback}: Consider a symmetric case where
$M_1=M_2=N_1=N_2=M$ and $D_{1,1}=D_{1,2}=D_{2,1}=D_{2,2}=D$. Here, we assume that $D$ is even, but we can get a similar graph when $D$ is odd. We plot
the total DoF as a function of $M$ with fixed $D$ in Fig.
\ref{fig2}. Note that the DoF gain due to feedback can be achieved
when the ratio of the number of antennas at each node to the rank of
each channel matrix is greater than a certain threshold. For
$M>1.5D$, we can achieve a higher DoF. The gain comes from the fact
that feedback can provide alternative signal paths when the number
of antennas at each node is large enough as compared to the channel
ranks. Here, we provide an intuition behind this
gain through a simple example.

\begin{figure}[!t]
\centering
\includegraphics[scale=0.7]{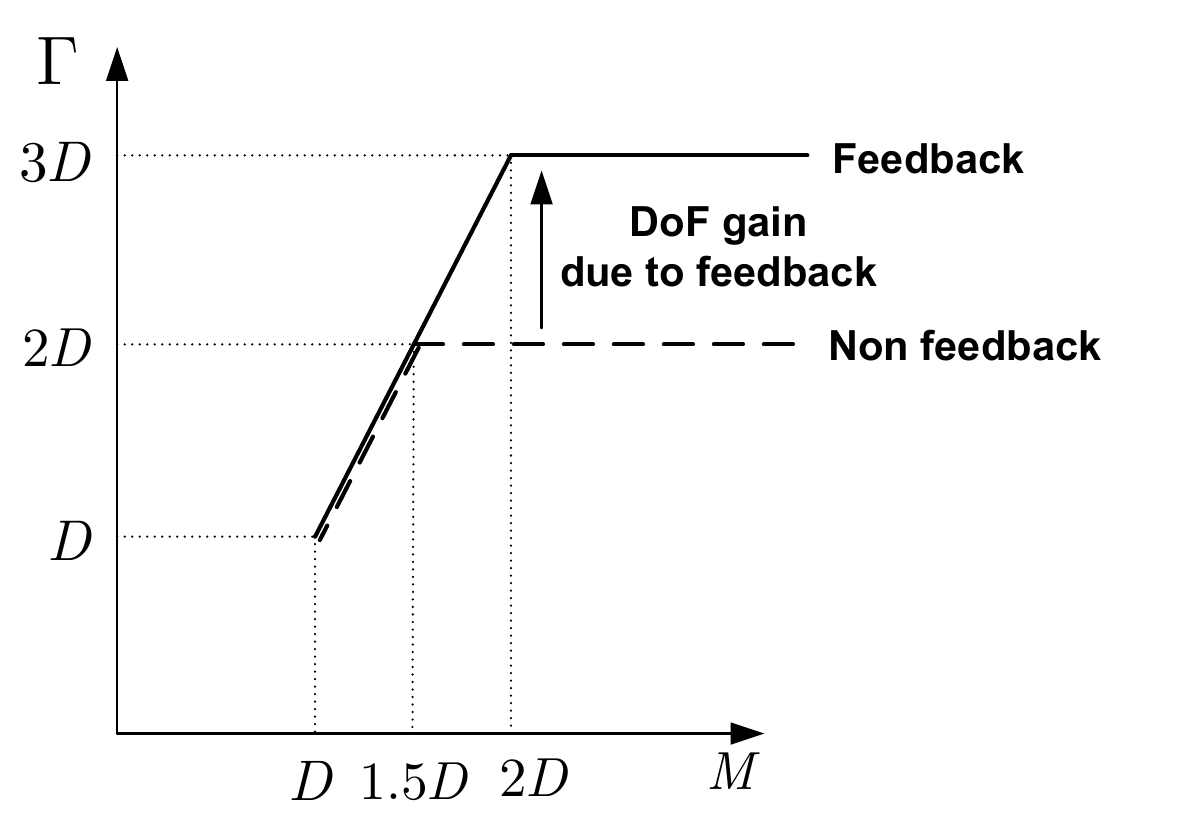}
\caption{Total DoF for two-user case when $M_1=M_2=N_1=N_2=M$ and
$D_{1,1}=D_{1,2}=D_{2,1}=D_{2,2}=D.$} \label{fig2} \vspace{-0.1in}
\end{figure}
\vspace{0.05in}

\begin{figure}[!t]
\centering
\includegraphics[scale=0.7]{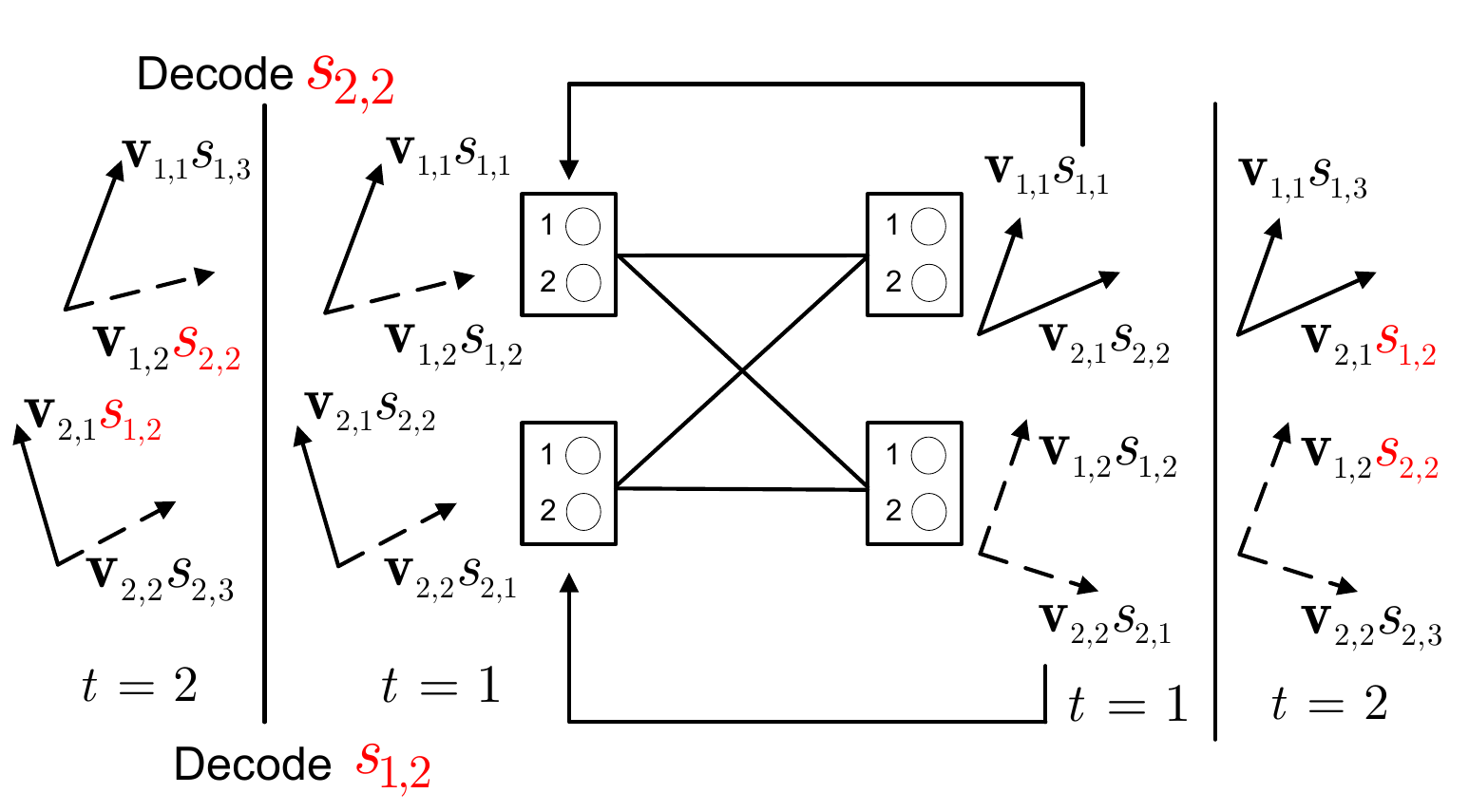}
\caption{Achievability in Example 1. The beamforming vectors
represented by solid and dashed lines denote the signals
transmitted to receivers $1$ and $2$, respectively.} \label{fig3}
\vspace{-0.1in}
\end{figure}
\vspace{0.05in}

 \vspace{0.05in}
\begin{example}
Consider the case where $M_1=M_2=N_1=N_2=2$ and
$D_{1,1}=D_{2,2}=D_{1,2}=D_{2,1}=1$. Our achievable scheme operates
in two time slots. See Fig. 3. Let $s_{i,j}$ denote the $j$th symbol
of user $i\in\{1,2\}$. At time slot 1, we design the transmitted
signals as
\begin{align*}
\mathbf{x}_1(1)&=\mathbf{v}_{1,1}s_{1,1}+\mathbf{v}_{1,2}s_{1,2}\\
\mathbf{x}_2(1)&=\mathbf{v}_{2,1}s_{2,2}+\mathbf{v}_{2,2}s_{2,1},
\end{align*}
where the beamforming vectors are designed such that
$\mathbf{H}_{2,1}\mathbf{v}_{1,1}=\mathbf{H}_{1,1}\mathbf{v}_{1,2}=\mathbf{H}_{2,2}\mathbf{v}_{2,1}=\mathbf{H}_{1,2}\mathbf{v}_{2,2}=0$.
Note that this design is feasible as the number of antennas at
each node is greater than the rank of each channel matrix. Then,
the received signal at each receiver is given by
\begin{align*}
\mathbf{y}_1(1)&=\mathbf{H}_{1,1}\mathbf{v}_{1,1}s_{1,1}+\mathbf{H}_{1,2}\mathbf{v}_{2,1}s_{2,2}+\mathbf{z}_1(1)\\
\mathbf{y}_2(1)&=\mathbf{H}_{2,2}\mathbf{v}_{2,2}s_{2,1}+\mathbf{H}_{2,1}\mathbf{v}_{1,2}s_{1,2}+\mathbf{z}_2(1).
\end{align*}
Here, $\text{rank}([\begin{array}{ccccc}
\mathbf{H}_{i,i}\mathbf{v}_{i,i}&  \mathbf{H}_{i,j}\mathbf{v}_{j,i} \\
\end{array}])=2$, $\forall i,j=1,2$ and $i\neq j$ with probability one since
channels are generic as $(1)$. Therefore, receiver 1 can decode
$s_{1,1}$ and transmitter 1 can know $s_{2,2}$ after receiving
$\mathbf{y}_1(1)$. Similarly, receiver $2$ and transmitter $2$ can
decode $s_{2,1}$ and $s_{1,2}$, respectively.

Now the idea is that at the next time slot, each transmitter sends
the other user's symbol in addition to its own fresh symbol. To
achieve this, we design the transmitted signals at time slot 2 as
\begin{align*}
\mathbf{x}_1(2)&=\mathbf{v}_{1,1}s_{1,3}+\mathbf{v}_{1,2}s_{2,2}\\
\mathbf{x}_2(2)&=\mathbf{v}_{2,2}s_{2,3}+\mathbf{v}_{2,1}s_{1,2},
\end{align*}
where $s_{1,3}$ and $s_{2,3}$ are new symbols for users $1$ and
$2$, respectively. Then we can see that receiver $1$ can decode
$(s_{1,2}, s_{1,3})$ and receiver $2$ can decode $(s_{2,2},
s_{2,3})$. As a result, six symbols can be transmitted over two
time slots, thus achieving $\Gamma_{fb}\geq 3$. This shows an
improvement over the non-feedback DoF of 2.
\end{example}

\begin{remark}
From the above example, we see that feedback can create new signal
paths (e.g., for $s_{1,2}$, transmitter $1\rightarrow$ receiver
$2\rightarrow$ feedback $\rightarrow$ transmitter $2\rightarrow$
receiver $1$), which do exist in the non-feedback case. When
the number of antennas at each node is large enough as compared to
the ranks of channel matrices, the dimension of signal space at
each node becomes sufficiently large such that some signals can be
transmitted through these new signal paths, thus increasing the
ranks of effective desired channel matrices. For instance, the
effective desired channel matrix for user $1$ at time slot $2$ is
given by $\mathbf{H}^{e}_{1,1}=\mathbf{H}_{1,1}+\mathbf{H}_{1,2}$,
where $\text{rank}(\mathbf{H}^{e}_{1,1})=2$. However, when all the
direct links have full ranks, feedback cannot increase the total
DoF since we cannot increase the ranks of direct links further and
cannot create such alternative signal paths. Note that the role of
feedback here is similar to that of relays in~\cite{Chae12}, which
shows that using multiple relays can create alternative signal
paths, thus increasing the total DoF in the rank-deficient
interference channel.
\end{remark}

\subsection{Three-user case}
When $K\geq 3$, we focus on a \emph{symmetric} case where
$M_i=N_i=M$, $D_{i,i}=D_d$, and $D_{j,i}=D_c$, $\forall
i=1,2,\ldots, K$ and $i\neq j$. Specifically, for $K=3$, we
develop a new achievable scheme which employs interference
alignment when the rank of cross links $D_c$ is sufficiently
large. The achievable total DoF for the three-user case is stated
in the following theorem.

\begin{theorem}[Lower bound for $K=3$]
For the symmetric three-user rank-deficient interference channel
with feedback, the following total DoF is achievable.
\begin{align*}
\Gamma_{fb}\geq \left\{\begin{array}{ll} \max\left\{\min\left\{\frac{3M}{2},M+D_d\right\},2M-D_c\right\}& \textrm{if $D_c\leq M\leq 2D_c$}, \\
 3M-3D_c& \textrm{if $2D_c\leq M\leq 2D_c+D_d$},\\
 3D_d+3D_c& \textrm{if $2D_c+D_d\leq M$}\\
\end{array} \right.
\end{align*}
\end{theorem}
\begin{proof}
See Section V for the proof.
\end{proof}
\vspace{0.05in}

\begin{remark}
If all the direct links have full ranks, i.e., $D_{d}=M$, the
result again recovers the non-feedback case
in~\cite{Krishnamurthy12}:
\begin{align*}
\Gamma_{no}\geq \left\{\begin{array}{ll}\frac{3M}{2}& \textrm{if $D_c\leq M\leq 2D_c$}, \\
 3M-3D_c& \textrm{if $2D_c\leq M$}.\\
\end{array} \right.
\end{align*}
\end{remark}
\vspace{0.05in}

\begin{remark}  As will be explained in Section V, our achievable scheme involves interference alignment with feedback when $D_c\leq M< 2D_c$, while it is merely based on zero-forcing when $M\geq 2D_c$.
This is due to the fact that when the ratio of $D_c$ to $M$ is
greater than a certain threshold, we cannot null out all the interference signals, and thus aligning unwanted signals is required to utilize the dimension of the received signal space more efficiently. Furthermore,
as will be shown in Theorem 3, the proposed scheme achieves the
optimal total DoF when $M\geq 2D_c+D_d$.
\end{remark}
\vspace{0.05in}

\begin{figure}[!t]
\centering
\includegraphics[scale=0.8]{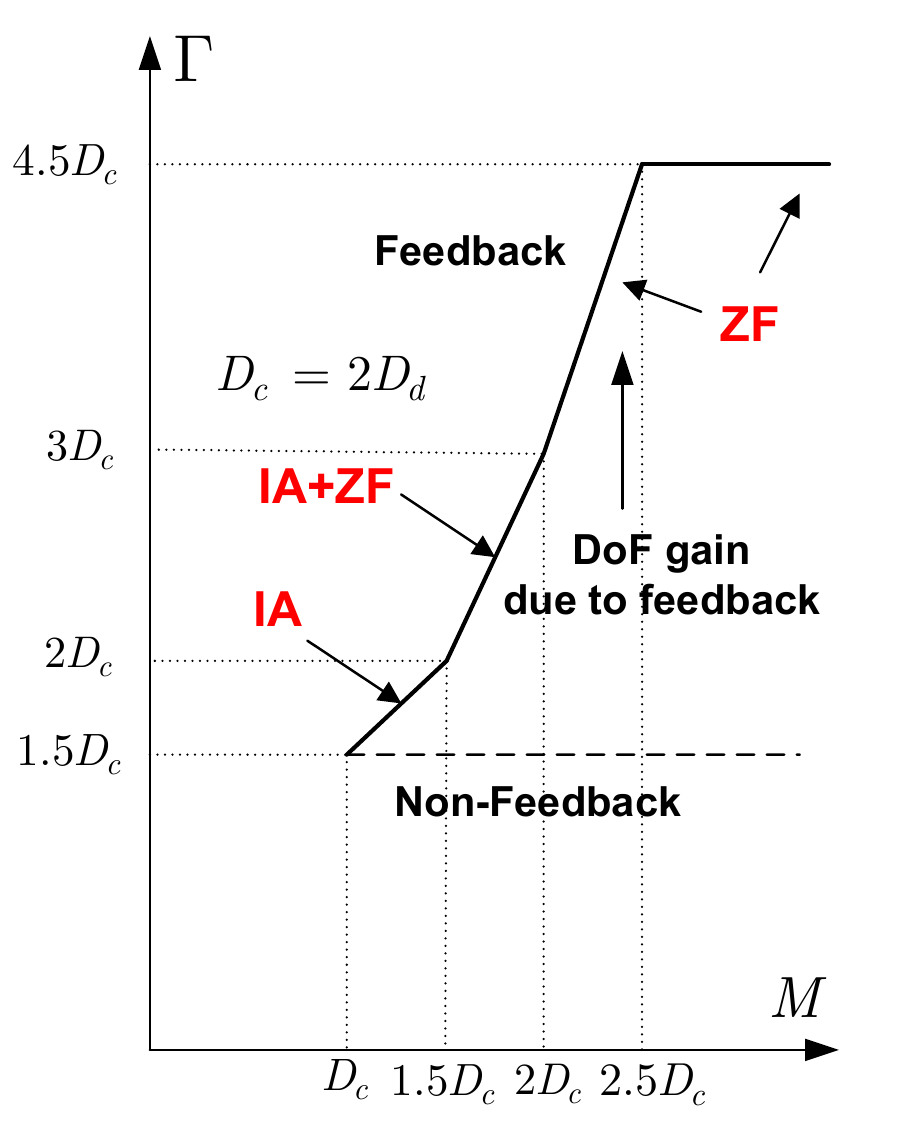}
\caption{Achievable total DoF for the three-user case when
$D_c=2D_d$. The achievable scheme is based on zero-forcing (ZF)
and/or interference alignment (IA), depending on the
number of antennas at each node.} \label{fig4} \vspace{-0.1in}
\end{figure}
\vspace{0.05in}

\textbf{DoF gain due to feedback}: Consider the case where
$D_c=2D_d$. Again, we plot the total DoF as a function of $M$ with
fixed $D_c$ and $D_d$ in Fig. \ref{fig4}. Note that for the
three-user case, we employ interference alignment when the
rank of each cross link $D_c$ is sufficiently large as compared to
the number of antennas at each node $M$ (Here, when $M\leq 2D_c$).
In addition, we can see that the slope in Fig.~\ref{fig4}
increases with the number of antennas. This is because if there
are enough antennas at each node, we can even create new
interference-free signal paths via zero-forcing rather than
aligning unwanted signals.

 \vspace{0.05in}
We provide a simple example that shows how interference alignment
can be applied with feedback. \vspace{0.05in}

\begin{figure}[!t]
\centering
\includegraphics[scale=0.68]{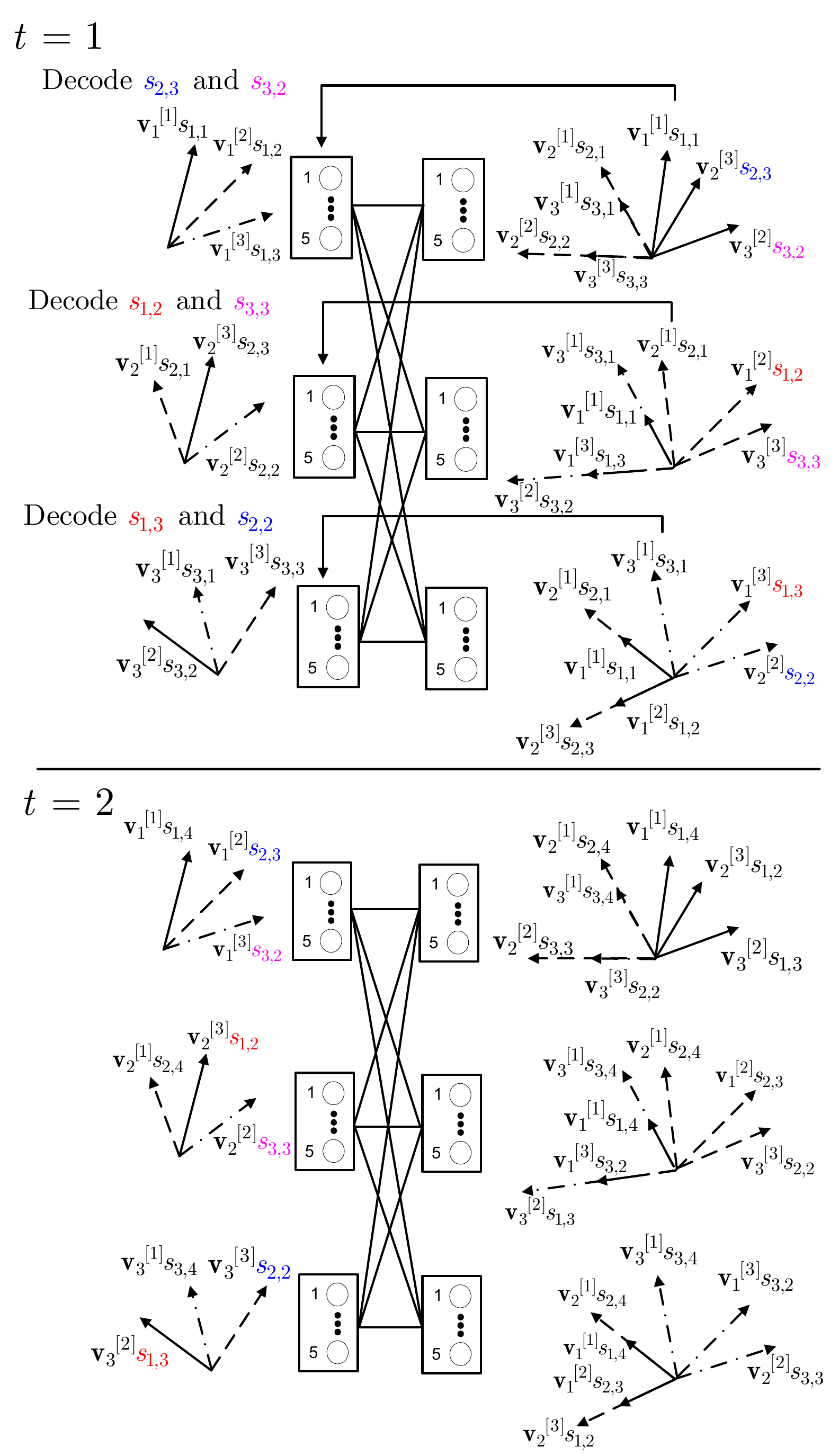}
\caption{Achievability in Example 2. The beamforming vectors
represented by solid, dashed, and dashed-dotted lines denote the
desired signals for receivers $1$, $2$, and $3$ respectively at each
time slot. Note that arrows in the figure represent linearly
independent directions in a five dimensional space.} \label{fig4_4}
\vspace{-0.1in}
\end{figure}

\begin{example}
Consider the case where $M=D_c=5$ and $D_d=1$. As in the two-user
case, the proposed scheme operates in two time slots. See
Fig.~\ref{fig4_4}.  At time slot 1, we design the transmitted
signal for transmitter $i\in\{1,2,3\}$ as
\begin{align*}
\mathbf{x}_i(1)=\mathbf{v}_i^{[1]}
s_{i,1}+\mathbf{v}_i^{[2]}s_{i,2}+\mathbf{v}_i^{[3]}s_{i,3}.
\end{align*}
Here, transmitter $i$ delivers $s_{i,1}$, $s_{i,2}$, and
$s_{i,3}$ to receivers $i$, $i+1$, and ${i+2}$, respectively, while aligning unwanted signals for each receiver. Note that although $s_{i,2}$ and $s_{i,3}$ are not intended symbols for receivers $i+1$ and $i+2$, using feedback, transmitters $i+1$ and $i+2$ will forward them to receiver $i$ in the next time slot.

To achieve
these, we construct $\mathbf{v}_i^{[1]}$ such that
\begin{align*}
\text{span}\left(\mathbf{H}_{i+1,i}\mathbf{v}^{[1]}_i\right)\subseteq
\text{span}\left(\mathbf{H}_{i+1,i+2}\mathbf{v}^{[1]}_{i+2}\right).
\end{align*}
We also design $\mathbf{v}_i^{[2]}$ and $\mathbf{v}_i^{[3]}$ such
that
\begin{align*}
&\mathbf{H}_{i,i}\mathbf{v}_{i}^{[3]}=\mathbf{H}_{i+2,i+2}\mathbf{v}_{i+2}^{[2]}=0,\\
&\mathbf{H}_{i+1,i}\mathbf{v}_{i}^{[3]}=\mathbf{H}_{i+1,i+2}\mathbf{v}_{i+2}^{[2]}.
\end{align*}
This beamforming design is feasible for $M\geq 2D_d$ and $M\leq
2D_c$ (This will be clarified in Section IV.). It turns out that
for receiver $i+1$, unwanted symbols ($s_{i,1}$, $s_{i+2,1}$) and
($s_{i,3}$, $s_{i+2,2}$) are aligned. Now we can see that
\begin{align*}
&\text{rank}\left(\left[\begin{array}{ccccccccc}
\mathbf{H}_{i,i+1}\mathbf{v}_{i+1}^{[1]} &
\mathbf{H}_{i,i+2}\mathbf{v}_{i+2}^{[1]}
\end{array}\right]\right)=1,\\
&\text{rank}\left(\left[\begin{array}{ccccccccc}
\mathbf{H}_{i,i+1}\mathbf{v}_{i+1}^{[2]} &
\mathbf{H}_{i,i+2}\mathbf{v}_{i+2}^{[3]}
\end{array}\right]\right)=1
\end{align*}
with probability one. Hence, receiver $i$ can decode $s_{i,1}$.

On the other hand, transmitter $i$ can know $s_{i+2,2}$ and
$s_{i+1,3}$ after receiving $\mathbf{y}_i(1)$. At the next time
slot, each transmitter forwards the other user's symbols in addition to its own fresh symbol. To achieve this, we design the transmitted signal
as
\begin{align*}
\mathbf{x}_i(2)=\mathbf{v}_i^{[1]}
s_{i,4}+\mathbf{v}_i^{[2]}s_{i+1,3}+\mathbf{v}_i^{[3]}s_{i+2,2},
\end{align*}
where $s_{i,4}$ is a new symbol for user $i$. Then, using the same
argument above, we can see that receiver $i$ can decode $(s_{i,2},
s_{i,3}, s_{i,4})$, $\forall i=1,2,3$. As a result, we can send 12
symbols over two time slots, thus achieving $\Gamma_{fb} \geq 6$.
Note that the total DoF becomes three when there is no feedback.
\end{example}

\subsection{Upper bound for the $K$-user case}

\begin{theorem} [Upper bound for the $K$-user case]

For the symmetric $K$-user rank-deficient interference channel with
feedback, the total DoF is upper bounded by
\begin{align*}
\Gamma_{fb} \leq  KD_d+\frac{D_cK(K-1)}{2}.
\end{align*}
\end{theorem}
\begin{proof}
See Section VI for the proof.
\end{proof}

\begin{cor}
For the symmetric $K$-user rank-deficient interference channel with feedback, the total DoF is given by
\begin{align*}
\Gamma_{fb} = KD_d+\frac{D_cK(K-1)}{2}
\end{align*}
when $M\geq D_d+(K-1)D_c$.
\end{cor}
\begin{proof}
The converse follows from Theorem 3. For achievability, we
consider a simple extension of the scheme in Theorem 2. At the
first time slot, each transmitter sends total $D_d+(K-1)D_c$
symbols, in which $D_d$ symbols are sent through the direct link
and $D_c$ symbols are sent through each cross link. Then, after
receiving $\mathbf{y}_i(1)$, transmitter $i$ and receiver $i$ can
know $D_d$ desired symbols and $(K-1)D_c$ the other user's symbols.
This is possible due to the fact that $M\geq D_d+(K-1)D_c$. At the
second time slot, each transmitter sends its new $D_d$ symbols and
also forwards the other user's symbols to the corresponding receivers.
Consequently, we can see that each receiver can decode
$2D_d+(K-1)D_c$ symbols during two-time slots, thus achieving
$\Gamma_{fb} \geq KD_d+\frac{D_cK(K-1)}{2}$.
\end{proof}
\begin{remark}
Suppose there are sufficiently many antennas at each node (e.g.,
$M\gg D_d+(K-1)D_c$.). Then, from Corollary 1, we can see that the
DoF gain due to feedback increases with the number of users. Let
$\Gamma_{fb}$ and $\Gamma_{no}$ denote the total DoFs when there
is feedback and no feedback, respectively. In addition, consider
the case where $D_c=D_d=D$. Then, we have
\begin{align*}
\frac{\Gamma_{fb}}{\Gamma_{no}}=\frac{DK(K+1)/2}{DK}=\frac{K+1}{2}
\end{align*} and we can see that the DoF gain is proportional to the number of users.
\end{remark}

\section{Proof of Theorem 1}
\subsection{Achievability}
Our achievable scheme operates in two time slots. For brevity, we
first categorize beamforming vectors for transmitter $i\in\{1,2\}$
into three types:
\begin{align*}
\mathbf{V}_i=\left[\begin{array}{ccccccccccccccc}
\mathbf{V}_i^{[1]}& \mathbf{V}_i^{[2]}& \mathbf{V}_i^{[3]}
\end{array}\right],
\end{align*}
where $\mathbf{V}_i^{[j]}$ is a concatenation of type $j$
beamforming vectors of transmitter $i$, i.e.,
$$\mathbf{V}_i^{[j]}=\left[\begin{array}{ccccccccccccccc}
\mathbf{v}^{[j]}_{i,{\sum^j_{k=1}d_i^{[k-1]}+1}}&
\mathbf{v}^{[j]}_{i,{\sum^j_{k=1}d_i^{[k-1]}+2}}& \cdots&
\mathbf{v}^{[j]}_{i,{\sum^j_{k=1}d_i^{k}}}
\end{array}\right],\quad \forall j=1,2,3,$$ $d_i^{[j]}$ denotes the number of vectors
in $\mathbf{V}_i^{[j]}$, and $d_i^{[0]}=0$.

\begin{itemize}
\item $\mathbf{v}_{i,k}^{[1]}$ denotes the $k$th beamforming
vector for transmitter $i$ which spans the nullspace of $\mathbf{H}_{j,i}$,
i.e., $\mathbf{H}_{j,i}\mathbf{v}_{i,k}^{[1]}=0$, and
$\mathbf{H}_{i,i}\mathbf{v}_{i,k}^{[1]}\neq 0$, where $i\neq j$.
Note that since $\text{rank}(\mathbf{H}_{j,i})=D_{j,i}$, the
maximum number of beamforming vectors satisfying this condition is
$M_i-D_{j,i}$.

\item $\mathbf{v}_{i,k}^{[2]}$ denotes the $k$th beamforming
vector for transmitter $i$ whose coefficients are randomly
generated from a continuous distribution and
$0<||\mathbf{v}_{i,k}^{[2]}||\leq A$, where $A$ is a finite value.
Hence, $\mathbf{H}_{i,i}\mathbf{v}_{i,k}^{[2]}\neq 0$ and
$\mathbf{H}_{j,i}\mathbf{v}_{i,k}^{[2]}\neq 0$ with probability
one.

\item $\mathbf{v}_{i,k}^{[3]}$ denotes the $k$th beamforming
vector for transmitter $i$ which spans the nullspace of  $\mathbf{H}_{i,i}$,
i.e., $\mathbf{H}_{i,i}\mathbf{v}_{i,k}^{[3]}=0$, and
$\mathbf{H}_{j,i}\mathbf{v}_{i,k}^{[3]}\neq 0$. Note that since
$\text{rank}(\mathbf{H}_{i,i})=D_{i,i}$, the maximum number of
beamforming vectors satisfying this condition is $M_i-D_{i,i}$.
\end{itemize}



Now we explain the proposed scheme. Let $\mathbf{s}_i^{[j]}(t)$
denote the symbols of user $i$ conveyed by $\mathbf{V}_{i}^{[j]}$
at time slot $t$. At the first time slot, we design the
transmitted signal as
\begin{align*}
\mathbf{x}_{i}(1)&=\mathbf{V}_i\mathbf{s}_i(1)\\
&=\mathbf{V}_{i}^{[1]}\mathbf{s}^{[1]}_{i}(1)+\mathbf{V}_{i}^{[2]}\mathbf{s}^{[2]}_{i}(1)+\mathbf{V}_{i}^{[3]}\mathbf{s}^{[3]}_{i}(1),
\end{align*}
where
\begin{align*}
\mathbf{s}_i(1)=\left[\begin{array}{ccccc}
\mathbf{s}^{[1]}_{i}(1)\\ \mathbf{s}^{[2]}_{i}(1)\\ \mathbf{s}^{[3]}_{i}(1)  \\
\end{array}\right],
\end{align*}
\begin{align*}
\mathbf{s}_{i}^{[j]}(1)=\left[\begin{array}{ccccc}
s_{i,\sum_{k=1}^{j}d_i^{[k-1]}+1}&
s_{i,\sum_{k=1}^{j}d_i^{[k-1]}+2} \cdots &
s_{i,\sum_{k=1}^{j}d_i^{k}}
\end{array}\right]^{T},\quad \forall j=1,2,3.
\end{align*}
Here, transmitters send their symbols with independent Gaussian
signaling, i.e.,
\begin{align*}
\mathbf{s}_i(1)~\sim
\mathcal{CN}\left(\mathbf{0}_{d},\frac{P}{d}\mathbf{I}_{d}\right)
\end{align*} where $d=d^{[1]}_i+d^{[2]}_i+d^{[3]}_i$, and
$\mathbf{V}_{i}^{[1]}$, $\mathbf{V}_{i}^{[2]}$, and
$\mathbf{V}_{i}^{[3]}$ are properly scaled to satisfy the power
constraint $P$. Then the received signal at receiver $i\in\{1,2\}$
is given by
\begin{align*}
\mathbf{y}_i(1)&=\mathbf{H}_{i,i}\mathbf{x}_i(1)+\mathbf{H}_{i,j}\mathbf{x}_j(1)+\mathbf{z}_i(1)\nonumber\\
\nonumber
&=\mathbf{H}_{i,i}\mathbf{V}_{i}^{[1]}\mathbf{s}^{[1]}_{i}(1)+\mathbf{H}_{i,i}\mathbf{V}_{i}^{[2]}\mathbf{s}^{[2]}_{i}(1)\\
&+\mathbf{H}_{i,j}\mathbf{V}_{j}^{[2]}\mathbf{s}^{[2]}_{j}(1)+\mathbf{H}_{i,j}\mathbf{V}_{j}^{[3]}\mathbf{s}^{[3]}_{j}(1)+\mathbf{z}_i(1).
\end{align*}

In the proposed scheme, we want to enable receiver $i$ to decode
its desired symbols $\mathbf{s}^{[1]}_{i}(1)$ and
$\mathbf{s}^{[2]}_{i}(1)$. In addition, we also want to make
transmitter $i$ be able to know the other user's symbols
$\mathbf{s}^{[2]}_{j}(1)$ and $\mathbf{s}^{[3]}_{j}(1)$ after its
corresponding receiver feeds back the received signal. To achieve
these, we choose $ d_{1}^{[1]}, d_{1}^{[2]}, d_{1}^{[3]},
d_{2}^{[1]}, d_{2}^{[2]}$, and $d_{2}^{[3]}$ to satisfy the
following conditions.
\begin{align}
&d_1^{[3]}=d_2^{[3]}\triangleq d^{f}\\
&0\leq d_{1}^{[1]}\leq M_1-D_{2,1}\\
&0\leq d_{2}^{[1]}\leq M_2-D_{1,2}\\
&0\leq d_f\leq \min\{M_1-D_{1,1},M_2-D_{2,2}\}\\
&0\leq d_{1}^{[1]}+d_{1}^{[2]}\leq D_{1,1}\\
&0\leq d_{2}^{[1]}+d_{2}^{[2]}\leq D_{2,2}\\
&0\leq d_{1}^{[2]}+d^{f}\leq D_{2,1}\\
&0\leq d_{2}^{[2]}+d^{f}\leq D_{1,2}\\
&0\leq d_{1}^{[2]}+d^{f}+d_{2}^{[1]}+d_{2}^{[2]}\leq N_2\\
&0\leq d_{2}^{[2]}+d^{f}+d_{1}^{[1]}+d_{1}^{[2]}\leq N_1
\end{align}
Here, the conditions $(3)$-$(5)$ are due to the properties of
$\mathbf{V}_{i}^{[1]}$ and $\mathbf{V}_{i}^{[3]}$; $(6)$-$(9)$ are
due to the fact that the number of symbols transmitted through a
channel is constrained by the rank of the channel matrix;
$(10)$-$(11)$ are due to the fact that the number of received
symbols at a receiver should be less than or equal to the number
of antennas at the receiver. Note that if the above conditions are
satisfied, we have
\begin{align*}&\text{rank}\left(
\left[\begin{array}{ccccc} \mathbf{H}_{i,i}\mathbf{V}^{[1]}_{i}&
\mathbf{H}_{i,i}\mathbf{V}^{[2]}_{i}&
\mathbf{H}_{i,j}\mathbf{V}_{j}^{[2]} &
\mathbf{H}_{i,j}\mathbf{V}_{j}^{[3]}
\end{array}\right]\right)\\
&=d_{i}^{[1]}+d_{i}^{[2]}+d_{j}^{[2]}+d^{f}.
\end{align*}
with probability one $\forall i=1,2$ and $i\neq j$. This is due to
the facts that $\mathbf{V}_1$ and $\mathbf{V}_2$ are full-rank
matrices and channel matrices are generic so that
$\mathbf{H}_{i,i}$ and $\mathbf{H}_{i,j}$ are random linear
transformations. In addition, since $d_i^{[1]}+d_i^{[2]}\leq
D_{i,i}$ and $d_j^{[2]}+d^f\leq \min\{D_{1,2},D_{2,1}\}$, linear
independence of signals is also preserved. Thus, by observing
$\mathbf{y}_i(1)$, receiver $i$ and transmitter $i$ can obtain the
desired results. Consequently, at time slot 1, receivers $1$ and
$2$ can decode $d_1^{[1]}+d_1^{[2]}$ and $d_2^{[1]}+d_2^{[2]}$
symbols, respectively.

Now we consider the proposed scheme in the second time slot.
Recall that transmitter $i$ can know the other user's symbols
$\mathbf{s}^{[2]}_{j}(1)$ and $\mathbf{s}^{[3]}_{j}(1)$ after
receiving feedback signal $\mathbf{y}_i(1)$. Among these symbols,
transmitter $i$ sends only $\mathbf{s}^{[3]}_{j}(1)$ for receiver
$j$ at the next time slot since symbols of
$\mathbf{s}^{[2]}_{j}(1)$ were already decoded by receiver $j$ at
the first time slot.
Hence, at the second time slot, we design the transmitted signal as
\begin{align*}
\mathbf{x}_{i}(2)=\mathbf{V}_{i}^{[1]}\mathbf{s}^{[1]}_{i}(2)+\mathbf{V}_{i}^{[2]}\mathbf{s}^{[2]}_{i}(2)+\mathbf{V}_{i}^{[3]}\mathbf{s}^{[3]}_{j}(1),
\end{align*}
where
\begin{align*}
&\mathbf{s}_{i}^{[1]}(2)=\left[\begin{array}{ccccc}
s_{i,d^{[1]}_i+d^{[2]}_i+d_i^{[3]}+1} &  \cdots &  s_{i,2d^{[1]}_i+d^{[2]}_i+d_i^{[3]}}   \\
\end{array}\right]^{T},\\
&\mathbf{s}_{i}^{[2]}(2)=\left[\begin{array}{ccccc}
s_{i,2d^{[1]}_i+d^{[2]}_i+d_i^{[3]}+1} &  \cdots &
s_{i,2d^{[1]}_i+2d^{[2]}_i+d_i^{[3]}}
\end{array}\right]^{T}.
\end{align*}
Here, $\mathbf{s}^{[1]}_{i}(2)$ and $\mathbf{s}^{[2]}_{i}(2)$ are
new symbols of user $i$ transmitted at the second time slot. As a
result, the received signal at receiver $i\in\{1,2\}$ is given by
\begin{align*}
\mathbf{y}_i(2)&=\mathbf{H}_{i,i}\mathbf{x}_i(2)+\mathbf{H}_{i,j}\mathbf{x}_j(2)+\mathbf{z}_i(2)\nonumber\\
\nonumber
&=\mathbf{H}_{i,i}\mathbf{V}_{i}^{[1]}\mathbf{s}^{[1]}_{i}(2)+\mathbf{H}_{i,i}\mathbf{V}_{i}^{[2]}\mathbf{s}^{[2]}_{i}(2)\\
&+\mathbf{H}_{i,j}\mathbf{V}_{j}^{[2]}\mathbf{s}^{[2]}_{j}(2)+\mathbf{H}_{i,j}\mathbf{V}_{j}^{[3]}\mathbf{s}^{[3]}_{i}(1)+\mathbf{z}_i(2).
\end{align*}
Then, using the same argument as above,receiver $i$ can decode all
the symbols $\mathbf{s}^{[1]}_{i}(2)$, $\mathbf{s}^{[2]}_{i}(2)$,
and $\mathbf{s}^{[3]}_{i}(1)$.

In summary, during two time slots, receivers $1$ and $2$ can
decode $2d_1^{[1]}+2d_1^{[2]}+d_1^{[3]}$ and
$2d_2^{[1]}+2d_2^{[2]}+d_2^{[3]}$ symbols, respectively.
Therefore, the achievable total DoF is
\begin{align*}
\Gamma_{fb}\geq &d_1^{[1]}+d_1^{[2]}+d_2^{[1]}+d_2^{[2]}+\frac{d_1^{[3]}+d_2^{[3]}}{2}\\
=&d_1^{[1]}+d_1^{[2]}+d_2^{[1]}+d_2^{[2]}+d_f.
\end{align*}

Finally, by evaluating the conditions $(2)$-$(11)$ using the
Fourier-Motzkin elimination, we get the desired bound:
\begin{align*}
\Gamma_{fb}\geq
&\min\{M_1+N_2-D_{2,1},M_2+N_1-D_{1,2},\\&D_{1,1}+D_{2,2}+D_{1,2},
D_{1,1}+D_{2,2}+D_{2,1},\\
&\min\{M_1,N_1\}+D_{2,2},\min\{M_2,N_2\}+D_{1,1}\}.
\end{align*}

\vspace{0.05in}
\begin{remark}
The achievable total DoF can also be established in an alternative
way. One implicit strategy is to employ Lemma $1$ in~\cite{Suh11}.
We can achieve the same DoF by setting $X_i=U_{if}+U_i+X_{ip}$
where $U_{if}=\mathbf{V}_{i}^{[3]}\mathbf{s}^{[3]}_{i}(1)$,
$U=(U_{1f},U_{2f})$,
$U_i=\mathbf{V}_{i}^{[2]}\mathbf{s}^{[2]}_{i}(1)$, and
$X_{ip}=\mathbf{V}_{i}^{[1]}\mathbf{s}^{[1]}_{i}(1)$, $\forall
i=1,2$.
\end{remark}
 \vspace{0.05in}
\subsection{Converse}
The proof is a direct extension of that in the two-user SISO
interference channel with feedback~\cite{Suh11}. Hence, we focus
on explaining the steps needed for the rank-deficient channel.

Starting with Fano's inequality, we get:
\begin{align*}
n(R_1+R_2-\epsilon_n)&\leq I(W_1;\mathbf{y}_1^n)+I(W_2;\mathbf{y}_2^n)\\
&\leq
I(W_1;\mathbf{y}_1^n,\mathbf{s}_1^n,W_2)+I(W_2;\mathbf{y}_2^n)
\end{align*}
where
$\mathbf{s}_1=\mathbf{H}_{2,1}\mathbf{x}_{1}+\mathbf{z}_2$ as
in~\cite{Suh11}. Hence, by following the same steps
in~\cite{Suh11}, we have
\begin{align}
R_1+R_2&\leq
h(\mathbf{y}_2)+h(\mathbf{y}_1|\mathbf{s}_1,\mathbf{x_2})-h(\mathbf{z}_1)-h(\mathbf{z}_2).
\end{align}
Now we evaluate the inequality $(12)$ with respect to the number
of antennas at each node and the rank of each channel matrix. From
$(12)$, we have
\begin{align*}
R_1+R_2&\leq
h(\mathbf{y}_2)+h(\mathbf{y}_1|\mathbf{s}_1,\mathbf{x_2})-h(\mathbf{z}_1)-h(\mathbf{z}_2)\\
&\leq
h(\mathbf{y}_2)+h(\mathbf{H}_{1,1}\mathbf{x}_1+\mathbf{z}_1|\mathbf{s}_1)-h(\mathbf{z}_1)-h(\mathbf{z}_2).
\end{align*}
Notice that
\begin{align*}
&h(\mathbf{y}_2)-h(\mathbf{z}_2)\leq \log \left|K^{G}_{\mathbf{y}_2}\right|\\
&h(\mathbf{H}_{1,1}\mathbf{x}_1+\mathbf{z}_1|\mathbf{s}_1)-h(\mathbf{z}_1)\leq
\log\frac{\left|K^{G}_{(\mathbf{H}_{1,1}\mathbf{x}_1+\mathbf{z}_1,\mathbf{s}_1)}\right|}{\left|K^{G}_{\mathbf{s}_1}\right|},
\end{align*}
where $K^G_\mathbf{x}$ denotes the covariance matrix of a Gaussian
random vector $\mathbf{x}$~\cite{LNIT,Papailiopoulos12}.
Straightforward computation gives
\begin{align*}
&\log\left|K^{G}_{\mathbf{y}_2}\right|\leq \min\{N_2,
D_{2,2}+D_{2,1}\}\log P+o(\log P)\\
&\log\frac{\left|K^{G}_{(\mathbf{H}_{1,1}\mathbf{x}_1+\mathbf{z}_1,\mathbf{s}_1)}\right|}{\left|K^{G}_{\mathbf{s}_1}\right|}\leq
\min\{M_1-D_{2,1}, D_{1,1}\}\log P+o(\log P).
\end{align*}

Therefore, we have
\begin{align}
\Gamma_{fb}&\leq \min\{N_2, D_{2,2}+D_{2,1}\}+\min\{M_1-D_{2,1},
D_{1,1}\}\nonumber \\ \nonumber
&=\min\{N_2+M_1-D_{2,1}, N_2+D_{1,1},\\
 &~~~~~~~~~~M_1+D_{2,2}, D_{2,2}+D_{2,1}+D_{1,1}\}.
\end{align}

By symmetry, we can also get the following upper bound:
\begin{align}
\Gamma_{fb}&\leq \min\{N_1+M_2-D_{1,2}, N_1+D_{2,2},\nonumber \\
 &~~~~~~~~~~M_2+D_{1,1}, D_{1,1}+D_{1,2}+D_{2,2}\}.
\end{align}

Combining (13) and (14), we get the desired bound:
\begin{align*}
\Gamma_{fb}\leq
&\min\{M_1+N_2-D_{2,1},M_2+N_1-D_{1,2},\\&D_{1,1}+D_{2,2}+D_{1,2},
D_{1,1}+D_{2,2}+D_{2,1},\\
&\min\{M_1,N_1\}+D_{2,2},\min\{M_2,N_2\}+D_{1,1}\}.
\end{align*}

\section{Proof of Theorem 2}
As in the two-user case, our achievable scheme operates in two
time slots. The achievable scheme employs interference alignment
when $D_c$ is sufficiently large. For this section, we categorize
beamforming vectors for transmitter $i\in\{1,2,3\}$ into seven
types:
\begin{align*}
\mathbf{V}_i=\left[\begin{array}{ccccccccccccccc}
\mathbf{V}_i^{[1]}& \mathbf{V}_i^{[2]}&\cdots &\mathbf{V}_i^{[7]}
\end{array}\right].
\end{align*}
Here, since we consider the symmetric channel, we set
$d_i^{[j]}=d^{[j]}$, $\forall i=1,2,3$.

\begin{itemize}
\item $\mathbf{v}_{i,k}^{[1]}$ denotes the $k$th beamforming
vector for transmitter $i$ which spans the nullspace of  $\mathbf{H}_{i+1,i}$,
i.e., $\mathbf{H}_{i+1,i}\mathbf{v}_{i,k}^{[1]}=0$, and
$\mathbf{H}_{i,i}\mathbf{v}_{i,k}^{[1]}\neq 0$ and
$\mathbf{H}_{i+2,i}\mathbf{v}_{i,k}^{[1]}\neq 0$. Note that since
$\text{rank}(\mathbf{H}_{i+1,i})=D_{c}$, the maximum number of
beamforming vectors satisfying this condition is $M-D_{c}$.

\item $\mathbf{v}_{i,k}^{[2]}$ denotes the $k$th beamforming
vector for transmitter $i$ which spans the nullspace of  $\mathbf{H}_{i+2,i}$,
i.e., $\mathbf{H}_{i+2,i}\mathbf{v}_{i,k}^{[2]}=0$, and
$\mathbf{H}_{i,i}\mathbf{v}_{i,k}^{[2]}\neq 0$ and
$\mathbf{H}_{i+1,i}\mathbf{v}_{i,k}^{[2]}\neq 0$.

\item $\mathbf{v}_{i,k}^{[3]}$ denotes the $k$th beamforming
vector for transmitter $i$ which spans the nullspace of  $
[\begin{array}{cc}\mathbf{H}_{i+1,i}&
\mathbf{H}_{i+2,i}\end{array}]$, i.e.,
$\mathbf{H}_{i+1,i}\mathbf{v}_{i,k}^{[3]}=0$ and
$\mathbf{H}_{i+2,i}\mathbf{v}_{i,k}^{[3]}=0$, and
$\mathbf{H}_{i,i}\mathbf{v}_{i,k}^{[3]}\neq 0$.
 Note that this
type of beamforming vector exits only when $M\geq 2D_c$. Assuming
$M\geq 2D_c$, the maximum number of beamforming vectors satisfying
this condition is $M-2D_{c}$.

\item After determining $\mathbf{V}^{[1]}_i$, $\forall i=1,2,3$,
we construct alignment beamforming vectors for each transmitter.
We design $\mathbf{V}^{[4]}_i$ to satisfy
\begin{align*}
\text{span}\left(\mathbf{H}_{i+1,i}\mathbf{V}^{[4]}_i\right)\subseteq
\text{span}\left(\mathbf{H}_{i+1,i+2}\mathbf{V}^I_{i+2}\right)
\end{align*}
where
\begin{align*}
\mathbf{V}^I_{i}=\left[\begin{array}{ccccccccccccccc}
\mathbf{V}_i^{[1]} & \mathbf{V}_i^{[4]}
\end{array}
\right].
\end{align*}
To construct feasible $\mathbf{V}_{i}^{[4]}$, we employ the
beamforming scheme in~\cite{Krishnamurthy12}, which is proposed
for the non-feedback channel (Set $\mathbf{V}_{i}^{[4]}=V_i^{A}$
in~\cite{Krishnamurthy12}).

\item Consider the case where $M\geq 2D_d$ and $M\leq 2D_c$. Let
$\mathbf{V}_{i}^{[5]}=\left[\begin{array}{cccccc}\mathbf{V}_{i}^{[5,1]}&\mathbf{V}_{i}^{[5,2]}\end{array}\right]$,
where
\begin{align*}
&\mathbf{V}^{[5,1]}_i=\left[\begin{array}{ccccccccccccccc}
\mathbf{v}^{[5]}_{i,\sum_{j=1}^{4}d^{[j]}+1}&\cdots&\mathbf{v}^{[5]}_{i,\sum_{j=1}^{4}d^{[j]}+d^{[5]}/{2}}
\end{array}\right],\\
&\mathbf{V}^{[5,2]}_i=\left[\begin{array}{ccccccccccccccc}
\mathbf{v}^{[5]}_{i,\sum_{j=1}^{4}d^{[j]}+d^{[5]}/{2}}&\cdots&\mathbf{v}^{[5]}_{i,\sum_{j=1}^{5}d^{[j]}}
\end{array}\right].
\end{align*}

We construct alignment beamforming vectors
$\mathbf{V}_{i}^{[5,1]}$ and $\mathbf{V}_{i}^{[5,2]}$ such that
\begin{align*}
&\mathbf{H}_{i,i}\mathbf{V}_{i}^{[5,1]}=\mathbf{H}_{i,i}\mathbf{V}_{i}^{[5,2]}=0\\
&\mathbf{H}_{i+2,i}\mathbf{V}_{i}^{[5,1]}=\mathbf{H}_{i+2,i+1}\mathbf{V}_{i+1}^{[5,2]},
\end{align*}
or equivalently,
\begin{align*}
\underbrace{\left[\begin{array}{cccccc}\mathbf{H}_{i+2,i}&-\mathbf{H}_{i+2,i+1}\\
\mathbf{H}_{i,i}& \mathbf{0}_{M\times M}\\
\mathbf{0}_{M\times M}&
\mathbf{H}_{i+1,i+1}\end{array}\right]}_{\mathbf{T}}\left[\begin{array}{cccccc}\mathbf{V}_{i}^{[5,1]}\\
\mathbf{V}_{i+1}^{[5,2]}
\end{array}\right]=0.
\end{align*}

Since $\mathbf{T}$ is the $3M\times 2M$ matrix whose rank is
$M+2D_d$, we can find feasible $\mathbf{V}_{i}^{[5,1]}$ and
$\mathbf{V}_{i+1}^{[5,2]}$, where $d^{[5]}\leq 2M-4D_d$. For the
case where $M\leq 2D_d$ or $M\geq 2D_c$, we set $d^{[5]}=0$.

\item $\mathbf{v}_{i,k}^{[6]}$ denotes the $k$th beamforming
vector for transmitter $i$ which spans the nullspace of  $
[\begin{array}{cc}\mathbf{H}_{i,i}&
\mathbf{H}_{i+1,i}\end{array}]$, i.e,
$\mathbf{H}_{i,i}\mathbf{v}_{i,k}^{[6]}=0$ and
$\mathbf{H}_{i+1,i}\mathbf{v}_{i,k}^{[6]}=0$, and
$\mathbf{H}_{i+2,i}\mathbf{v}_{i,k}^{[6]}\neq 0$. Note that this
type of beamforming vector exits only when $M\geq D_d+D_c$.
Assuming $M\geq D_d+D_c$, the maximum number of beamforming
vectors satisfying this condition is $M-D_d-D_c$.

\item $\mathbf{v}_{i,k}^{[7]}$ denotes the $k$th beamforming
vector for transmitter $i$ which spans the nullspace of $[\begin{array}{cc}\mathbf{H}_{i,i}&
\mathbf{H}_{i+2,i}\end{array}]$, i.e,
$\mathbf{H}_{i,i}\mathbf{v}_{i,k}^{[7]}=0$ and
$\mathbf{H}_{i+2,i}\mathbf{v}_{i,k}^{[7]}=0$, and
$\mathbf{H}_{i+1,i}\mathbf{v}_{i,k}^{[7]}\neq 0$.
\end{itemize}
Notice that $\mathbf{V}_i^{[4]}$ and $\mathbf{V}_i^{[5]}$ are
alignment beamforming matrices while the others are zero-forcing
beamforming matrices.

Now we explain the proposed scheme. At time slot $t$, we design the
transmitted signal for transmitter $i\in{\{1,2,3}\}$ as
\begin{align*}
\mathbf{x}_{i}(t)&=\sum_{j=1}^{7}\mathbf{V}_i^{[j]}\mathbf{s}_i^{[j]}(t).
\end{align*}
Here, as in the two-user case, transmitters send their symbols with
independent Gaussian signaling, and beamforming vectors are properly
scaled to satisfy the power constraint. Then, due to the properties
of $\mathbf{V}_i^{[j]}$, $j\in\{1,2,\ldots, 7\}$, the received
signal of receiver $i$ at time slot 1 is given by
\begin{align}
\mathbf{y}_i(1)&=\mathbf{H}_{i,i}\mathbf{x}_i(1)+\mathbf{H}_{i,i+1}\mathbf{x}_{i+1}(1)+\mathbf{H}_{i,i+2}\mathbf{x}_{i+2}(1)+\mathbf{z}_i(1)\nonumber\\
&=\mathbf{H}_{i,i}\left(\mathbf{V}_{i}^{[1]}\mathbf{s}^{[1]}_{i}(1)+\mathbf{V}_{i}^{[2]}\mathbf{s}^{[2]}_{i}(1)+\mathbf{V}_{i}^{[3]}\mathbf{s}^{[3]}_{i}(1)+\mathbf{V}_{i}^{[4]}\mathbf{s}^{[4]}_{i}(1)\right)\nonumber \\
&+\mathbf{H}_{i,i+1}\left(\mathbf{V}_{i+1}^{[1]}\mathbf{s}^{[1]}_{i+1}(1)+\mathbf{V}_{i+1}^{[4]}\mathbf{s}^{[4]}_{i+1}(1)+\mathbf{V}_{i+1}^{[5]}\left[\begin{array}{ccccc}
\mathbf{s}_{i+1}^{[5,1]}(1) & \mathbf{s}_{i+1}^{[5,2]}(1)
\end{array}\right]^{T}+\mathbf{V}_{i+1}^{[6]}\mathbf{s}^{[6]}_{i+1}(1)\right)\nonumber\\
&+\mathbf{H}_{i,i+2}\left(\mathbf{V}_{i+2}^{[2]}\mathbf{s}^{[2]}_{i+2}(1)+\mathbf{V}_{i+2}^{[4]}\mathbf{s}^{[4]}_{i+2}(1)+\mathbf{V}_{i+2}^{[5]}\left[\begin{array}{ccccc}
\mathbf{s}_{i+2}^{[5,1]}(1) & \mathbf{s}_{i+2}^{[5,2]}(1)
\end{array}\right]^{T}+\mathbf{V}_{i+2}^{[7]}\mathbf{s}^{[7]}_{i+2}(1)\right)+\mathbf{z}_i(1),
\end{align}
where
\begin{align*}
\mathbf{s}_{i}^{[j]}(1)=\left[\begin{array}{ccccc}
s_{i,\sum_{k=1}^{j}d^{[k-1]}+1}& s_{i,\sum_{k=1}^{j}d^{[k-1]}+2}&
\cdots & s_{i,\sum_{k=1}^{j}d^{k}}
\end{array}\right]^{T},\quad \forall j=1,2,\ldots,7,
\end{align*}
\begin{align*}
&\mathbf{s}_{i}^{[5,1]}(1)=\left[\begin{array}{ccccc}
s_{i,\sum_{k=1}^{5}d^{[k-1]}+1}&
s_{i,\sum_{k=1}^{5}d^{[k-1]}+2}& \cdots &
s_{i,\sum_{k=1}^{5}d^{[k-1]}+d^{[5]}/2}\end{array}\right],\\
&\mathbf{s}_{i}^{[5,2]}(1)=\left[\begin{array}{ccccc}
s_{i,\sum_{k=1}^{5}d^{[k-1]}+d^{[5]}/2+1}&
s_{i,\sum_{k=1}^{5}d^{[k-1]}+d^{[5]}/2+2}& \cdots &
s_{i,\sum_{k=1}^{5}d^{k}}\end{array}\right].
\end{align*}
As we mentioned above, our achievable scheme operates in two time slots.
In the first time slot, transmitter $i$ delivers the
symbols $(\mathbf{s}_i^{[1]}(1),
\mathbf{s}_i^{[2]}(1),\mathbf{s}_i^{[3]}(1),\mathbf{s}_i^{[4]}(1))$,
$(\mathbf{s}_i^{[5,1]}(1), \mathbf{s}_i^{[7]}(1))$, and
$(\mathbf{s}_i^{[5,2]}(1), \mathbf{s}_i^{[6]}(1))$ to receivers
$i$, $i+1$, and $i+2$, respectively, where
\begin{align*}
&\mathbf{s}_{i}^{[5,1]}(1)=\left[\begin{array}{ccccc}
s_{i,\sum_{k=1}^{5}d^{[k-1]}+1}&\cdots &
s_{i,\sum_{k=1}^{5}d^{[k-1]}+d^{[5]}/2}\end{array}\right],\\
&\mathbf{s}_{i}^{[5,2]}(1)=\left[\begin{array}{ccccc}
s_{i,\sum_{k=1}^{5}d^{[k-1]}+d^{[5]}/2+1}&
\cdots &
s_{i,\sum_{k=1}^{5}d^{k}}\end{array}\right].
\end{align*}
Here, although $(\mathbf{s}_i^{[5,1]}(1), \mathbf{s}_i^{[7]}(1))$ and
$(\mathbf{s}_i^{[5,2]}(1), \mathbf{s}_i^{[6]}(1))$ are not desired symbols for receivers $i+1$ and $i+2$, using feedback, transmitters $i+1$ and $i+2$ will forward them to receiver $i$ in the next time slot. To achieve this, we choose $
d^{[j]}$ to satisfy the following conditions.
\begin{align}
&d^{[1]}\leq M-D_c\\
&d^{[2]}\leq M-D_c\\
&d^{[3]}\leq \max(M-2D_c,0)\\
&d^{[5]}\leq \max(2M-4D_d,0,\min(M-2D,0))\\
&d^{[6]}\leq \max(M-D_c-D_d,0)\\
&d^{[7]}\leq \max(M-D_c-D_d,0)\\
&d^{[1]}+d^{[2]}+d^{[3]}+d^{[4]}\leq D_d\\
&d^{[1]}+d^{[4]}+d^{[5]}+d^{[6]}\leq D_c\\
&d^{[2]}+d^{[4]}+d^{[5]}+d^{[7]}\leq D_c\\
&2d^{[1]}+2d^{[2]}+d^{[3]}+2d^{[4]}+\frac{3}{2}d^{[5]}+d^{[6]}+d^{[7]}\leq M\\
&\sum_{j=1}^{7}d^{[j]}\leq M
\end{align}
Here, the conditions $(16)$-$(21)$ are due to the properties of
$\mathbf{V}_{i}^{[j]}$; $(22)$-$(24)$ are due to the fact that the
number of symbols transmitted through a channel is constrained by
the rank of the channel matrix; $(25)$ is due to the fact that the
number of received symbols at a receiver should be less than or
equal to the number of antennas at the receiver; $(26)$ is due to
the fact that the number of transmitted symbols from a transmitter
should be less than or equal to the number of antennas at the
transmitter. Note that, due to the alignment properties of $\mathbf{V}_i^{[4]}$ and $\mathbf{V}_i^{[5]}$, we have
\begin{align*}
&\text{rank}\left(\left[\begin{array}{ccccccccc}\mathbf{H}_{i,i+1}\mathbf{V}_{i+1}^{[1]}&\mathbf{H}_{i,i+1}\mathbf{V}_{i+1}^{[4]}
& \mathbf{H}_{i,i+2}\mathbf{V}_{i+2}^{[4]}
\end{array}\right]\right)=d^{[1]}+d^{[4]},\\
&\text{rank}\left(\left[\begin{array}{ccccccccc}\mathbf{H}_{i,i+1}\mathbf{V}_{i+1}^{[5]}&\mathbf{H}_{i,i+2}\mathbf{V}_{i+2}^{[5]}
\end{array}\right]\right)=\frac{3}{2}d^{[5]}.
\end{align*}

Then, we have
\begin{align}
&\text{rank}\left(\mathbf{A}_1\right)=\sum_{j=1}^{4}d^{[j]},\\
&\text{rank}\left(\mathbf{A}_2\right)=d^{[1]}+d^{[4]}+d^{[5]}+d^{[6]},\\
&\text{rank}\left(\mathbf{A}_3\right)=d^{[2]}+d^{[4]}+d^{[5]}+d^{[7]},\\
&\text{rank}\left(\left[\begin{array}{ccccccccc}\mathbf{A}_2 &
\mathbf{A}_3
\end{array}\right]\right)=d^{[1]}+d^{[2]}+d^{[4]}+\frac{3}{2}d^{[5]}+d^{[6]}+d^{[7]},\\
&\text{rank}\left(\left[\begin{array}{ccccccccc}\mathbf{A}_1 &
\mathbf{A}_2 & \mathbf{A}_3
\end{array}\right]\right)=2d^{[1]}+2d^{[2]}+d^{[3]}+2d^{[4]}+\frac{3}{2}d^{[5]}+d^{[6]}+d^{[7]}
\end{align}
with probability one, where
\begin{align*}
\mathbf{A}_1&=\left[\begin{array}{ccccccccc}\mathbf{H}_{i,i}\mathbf{V}_{i}^{[1]}&\mathbf{H}_{i,i}\mathbf{V}_{i}^{[2]}
& \mathbf{H}_{i,i}\mathbf{V}_{i}^{[3]} &
\mathbf{H}_{i,i}\mathbf{V}_{i}^{[4]}
\end{array}\right],\\
\mathbf{A}_2&=\left[\begin{array}{ccccccccc}\mathbf{H}_{i,i+1}\mathbf{V}_{i+1}^{[1]}&\mathbf{H}_{i,i+1}\mathbf{V}_{i+1}^{[4]}
& \mathbf{H}_{i,i+1}\mathbf{V}_{i+1}^{[5]} &
\mathbf{H}_{i,i+1}\mathbf{V}_{i+1}^{[6]}
\end{array}\right],\\
\mathbf{A}_3&=\left[\begin{array}{ccccccccc}\mathbf{H}_{i,i+2}\mathbf{V}_{i+2}^{[2]}&\mathbf{H}_{i,i+2}\mathbf{V}_{i+2}^{[4]}
& \mathbf{H}_{i,i+2}\mathbf{V}_{i+2}^{[5]} &
\mathbf{H}_{i,i+2}\mathbf{V}_{i+2}^{[7]}
\end{array}\right].
\end{align*}
Notice that (27)-(31) are due to the facts that $\mathbf{V}_{1}$,
$\mathbf{V}_{2}$, and $\mathbf{V}_{3}$ are full-rank matrices and
all the channel matrices are generic. Thus, by observing
$\mathbf{y}_i(1)$, receiver $i$ and transmitter $i$ can decode the
desired symbols $(\mathbf{s}_i^{[1]}(1),
\mathbf{s}_i^{[2]}(1),\mathbf{s}_i^{[3]}(1),\mathbf{s}_i^{[4]}(1))$
and the other user's symbols $(\mathbf{s}_{i+1}^{[5,2]}(1),
\mathbf{s}_{i+1}^{[6]}(1), \mathbf{s}_{i+2}^{[5,1]}(1),
\mathbf{s}_{i+2}^{[7]}(1) )$ as desired.

Now we consider the proposed scheme in the second time slot.
Recall that transmitter $i$ can know the other user's symbols
$(\mathbf{s}_{i+1}^{[5,2]}(1), \mathbf{s}_{i+1}^{[6]}(1),
\mathbf{s}_{i+2}^{[5,1]}(1), \mathbf{s}_{i+2}^{[7]}(1))$ after
receiving $\mathbf{y}_i(1)$. In the second time slot, each
transmitter will forward these symbols to the corresponding
receivers, i.e.,  forward $(\mathbf{s}_{i+1}^{[5,2]}(1),
\mathbf{s}_{i+1}^{[6]}(1))$ to receiver $i+1$ and
$(\mathbf{s}_{i+2}^{[5,1]}(1), \mathbf{s}_{i+2}^{[7]}(1))$ to
receiver $i+2$, and also send its own new symbols. To achieve
this, we set the symbols of user $i$ transmitted at time slot 2 as
\begin{align*}
&\mathbf{s}_{i}^{[j]}(2)=\left[\begin{array}{ccccc}
s_{i,\sum_{k=1}^{j}d^{[k-1]}+\sum_{l=1}^{7}d^{[l]}+1}&
\cdots & s_{i,\sum_{k=1}^{j}d^{k}+\sum_{l=1}^{7}d^{[l]}}
\end{array}\right]^{T},\quad \forall j=1, \ldots, 4,\\
&\mathbf{s}_{i}^{[5]}(2)=\left[\begin{array}{ccccc}
\mathbf{s}_{i+1}^{[5,2]}(1) & \mathbf{s}_{i+2}^{[5,1]}(1)
\end{array}\right]^{T},\\
&\mathbf{s}_{i}^{[6]}(2)=\mathbf{s}_{i+2}^{[7]}(1),\\
&\mathbf{s}_{i}^{[7]}(2)=\mathbf{s}_{i+1}^{[6]}(1).
\end{align*}
Here, $\mathbf{s}_i^{[1]}$(2), $\mathbf{s}_i^{[2]}$(2), $\mathbf{s}_i^{[3]}$(2), and $\mathbf{s}_i^{[4]}$(2) are new symbols of user $i$ transmitted at the second time slot.
As a result, the received signal of receiver $i$ at time slot 2 is given by
\begin{align*}
\mathbf{y}_i(2)&=\mathbf{H}_{i,i}\left(\mathbf{V}_{i}^{[1]}\mathbf{s}^{[1]}_{i}(2)+\mathbf{V}_{i}^{[2]}\mathbf{s}^{[2]}_{i}(2)+\mathbf{V}_{i}^{[3]}\mathbf{s}^{[3]}_{i}(2)+\mathbf{V}_{i}^{[4]}\mathbf{s}^{[4]}_{i}(2)\right)\nonumber \\
&+\mathbf{H}_{i,i+1}\left(\mathbf{V}_{i+1}^{[1]}\mathbf{s}^{[1]}_{i+1}(2)+\mathbf{V}_{i+1}^{[4]}\mathbf{s}^{[4]}_{i+1}(2)+\mathbf{V}_{i+1}^{[5]}\left[\begin{array}{ccccc}
\mathbf{s}_{i+2}^{[5,2]}(1) & \mathbf{s}_{i}^{[5,1]}(1)
\end{array}\right]^{T}+\mathbf{V}_{i+1}^{[6]}\mathbf{s}^{[7]}_{i+1}(1)\right)\nonumber\\
&+\mathbf{H}_{i,i+2}\left(\mathbf{V}_{i+2}^{[2]}\mathbf{s}^{[2]}_{i+2}(2)+\mathbf{V}_{i+2}^{[4]}\mathbf{s}^{[4]}_{i+2}(2)+\mathbf{V}_{i+2}^{[5]}\left[\begin{array}{ccccc}
\mathbf{s}_{i}^{[5,2]}(1) & \mathbf{s}_{i+1}^{[5,1]}(1)
\end{array}\right]^{T}+\mathbf{V}_{i+2}^{[7]}\mathbf{s}^{[6]}_{i+2}(2)\right)+\mathbf{z}_i(2).
\end{align*}
Then, using the same argument as above, receiver $i$ can decode
all the symbols ($\mathbf{s}^{[1]}_{i}(2)$,
$\mathbf{s}^{[2]}_{i}(2)$,  $\mathbf{s}^{[3]}_{i}(2)$,
$\mathbf{s}^{[4]}_{i}(2)$,  $\mathbf{s}^{[5]}_{i}(1)$,
$\mathbf{s}^{[6]}_{i}(1)$, $\mathbf{s}^{[7]}_{i}(1)$).

In summary, during two time slots, receiver $i\in\{1,2,3\}$ can
decode
$2d^{[1]}+2d^{[2]}+2d^{[3]}+2d^{[4]}+d^{[5]}+d^{[6]}+d^{[7]}$
desired symbols, thus achieving total DoF:
\begin{align*}
\Gamma_{fb}\geq
d^{[1]}+d^{[2]}+d^{[3]}+d^{[4]}+\frac{d^{[5]}+d^{[6]}+d^{[7]}}{2}.
\end{align*}


Now we analyze the achievable total DoF by determining suitable
$d^{[j]}$ for all $j=1,2,\ldots,7$ with respect to $M$, $D_d$ and
$D_c$.

\subsection{Case 1 : when $D_c\leq M\leq 2D_c$}
\subsubsection{$M\geq 2D_d$ and $M\leq D_d+D_c$}
In this case, the proposed scheme employs interference alignment.
We set\footnote{If $\frac{2M-4D_d}{3}$ is not an integer, we
consider the three-time symbol extension as
in~\cite{Cadambe107},\cite{Krishnamurthy12}. Furthermore, whenever
$d^{[j]}$ is not an integer, we can consider a proper symbol
extension.}
\begin{align*}
&d^{[1]}=M-D_c,  \\
&d^{[4]}=D_d+D_c-M, \\
&d^{[2]}=d^{[3]}=d^{[6]}=d^{[7]}=0,\\
&d^{[5]}=\frac{2M-4D_d}{3},
\end{align*}
which satisfies the conditions (16)-(26). Then, during two time
slots, receiver $i\in\{1,2,3\}$ can decode
$2d^{[1]}+2d^{[4]}+d^{[5]}=2D_d+\frac{2M-4D_d}{3}$ symbols, thus
achieving the following total DoF:
\begin{align}
\Gamma_{fb}&\geq 3\left(D_d+\frac{M-2D_d}{3}\right)\nonumber \\
&=M+D_d.
\end{align}

\subsubsection{When $M\geq 2D_d$ and $M\geq D_d+D_c$} As in the previous case, the proposed scheme involves interference alignment. We set
\begin{align*}
&d^{[1]}=d^{[2]}=\frac{D_d}{2}, \\
&d^{[3]}=d^{[4]}=0,\\
&d^{[5]}=\frac{4D_c-2M}{3},\\
&d^{[6]}=d^{[7]}=M-D_c-D_d,
\end{align*}
which satisfies the conditions (16)-(26). Then, during two time
slots, receiver $i\in\{1,2,3\}$ can decode
$2d^{[1]}+2d^{[2]}+d^{[5]}+d^{[6]}+d^{[7]}$ symbols, thus
achieving the total DoF:
\begin{align}
\Gamma_{fb}&\geq 3\left(M-D_c+\frac{2D_c-M}{3}\right)\nonumber \\
&=2M-D_c.
\end{align}

\subsubsection{When $M\leq 2D_d$} In this case, we use the non-feedback
scheme in~\cite{Krishnamurthy12} and can achieve
\begin{align}
\Gamma_{fb}\geq \frac{3M}{2}
\end{align}
by setting
\begin{align*}
&d^{[1]}=M-D_c, \\
&d^{[4]}=D_c-\frac{M}{2}, \\
&d^{[2]}=d^{[3]}=d^{[5]}=d^{[6]}=d^{[7]}=0.
\end{align*}

Combining (32), (33), and (34), we obtain the following lower
bound on the total DoF.
\begin{align}
\Gamma_{fb}\geq
\max\left\{\min\left\{\frac{3M}{2},M+D_d\right\},2M-D_c\right\},\quad
\textrm{if $D_c\leq M\leq 2D_c$}.
\end{align}

\subsection{Case 2 : when $2D_c\leq M\leq 2D_c+D_d$}
\subsubsection{When $M\geq D_c+D_d$}
In this case, the proposed scheme is merely based on zero forcing
$(d^{[4]}=d^{[5]}=0)$. We set
\begin{align*}
&d^{[1]}=d^{[2]}=\frac{2D_c+D_d-M}{2}, \\
&d^{[3]}=M-2D_c,\\
&d^{[4]}=d^{[5]}=0,\\
&d^{[6]}=d^{[7]}=M-D_c-D_d
\end{align*}
which satisfies the conditions (16)-(26). Then, the achievable
total DoF is given by
\begin{align}
\Gamma_{fb}&\geq 3\left(D_d+M-D_c-D_d\right)\nonumber \\
&=3M-3D_c.
\end{align}
\subsubsection{When $M\leq D_d+D_c$}
 In this case, we use the non-feedback
scheme in~\cite{Krishnamurthy12} and can achieve
\begin{align}
\Gamma_{fb}\geq 3M-3D_c
\end{align}
by setting
\begin{align*}
&d^{[1]}=d^{[2]}=\frac{D_c}{2}, \\
&d^{[3]}=M-2D_c,\\
&d^{[4]}=d^{[5]}=d^{[6]}=d^{[7]}=0.
\end{align*}

Combining (36) and (37), we obtain the following lower bound on
the total DoF.
\begin{align}
\Gamma_{fb}\geq 3M-3D_c, \quad \textrm{if $2D_c\leq M\leq
2D_c+D_d$}.
\end{align}

\subsection{Case 3 : when $M\geq 2D_c+D_d$}
In this case, we use only $2D_c+D_d$ antennas out of $M$ antennas
at each node. Then, from the result in Case 2, we can achieve
\begin{align}
\Gamma_{fb}\geq 3D_d+3D_c, \quad \textrm{if $M\geq 2D_c+D_d$},
\end{align}
by setting
\begin{align*}
&d^{[1]}=d^{[2]}=d^{[4]}=d^{[5]}=0 \\
&d^{[3]}=D_d\\
&d^{[6]}=d^{[7]}=D_c.
\end{align*}

Finally, by combining (35), (38), and (39), we get the desired
bound:
\begin{align*}
\Gamma_{fb}\geq \left\{\begin{array}{ll} \max\left\{\min\left\{\frac{3M}{2},M+D_d\right\},2M-D_c\right\}& \textrm{if $D_c\leq M\leq 2D_c$}, \\
 3M-3D_c& \textrm{if $2D_c\leq M\leq 2D_c+D_d$},\\
 3D_d+3D_c& \textrm{if $2D_c+D_d\leq M$}\\
\end{array} \right.
\end{align*}

\section{Proof of Theorem 3}
Let $\bar{W}_{i}\triangleq\{W_{i+1},W_{i+2},\ldots,W_K\}$,
$\bar{X}_{i}\triangleq\{\mathbf{x}^n_{i+1},\mathbf{x}^n_{i+2},\ldots,
\mathbf{x}^n_{K}\}$, and
$\bar{Y}_{i}\triangleq\{\mathbf{y}^n_{i+1},\mathbf{y}^n_{i+2},\ldots,
\mathbf{y}^n_{K}\}$ $\forall i=1,2,\ldots,K$, where
$\bar{W}_{K}=\bar{X}_{K}=\bar{Y}_{K}=\emptyset$. Starting from
Fano's inequality, we have
\begin{align*}
n\left(\sum_{i=1}^{K}R_i-\epsilon_n\right)&\leq
\sum_{i=1}^{K}I\left(W_i;\mathbf{y}_i^n\right)\\
&\stackrel{(a)}{\leq}  \sum_{i=1}^{K}I\left(W_i;\mathbf{y}_i^n,\bar{W}_{i},\bar{Y}_{i}\right)\\
&\stackrel{(b)}{=} \sum_{i=1}^{K} I\left(W_i;\mathbf{y}_i^n,\bar{Y}_{i}|\bar{W}_{i}\right)\\
&=\sum_{i=1}^{K}
h\left(\mathbf{y}_i^n,\bar{Y}_{i}|\bar{W}_{i}\right)-h\left(\mathbf{y}_i^n,\bar{Y}_{i}|\bar{W}_{i},W_i\right)\\
&=\sum_{i=1}^{K}
h\left(\mathbf{y}_i^n|\bar{Y}_{i},\bar{W}_{i}\right)+\sum_{i=1}^{K}h\left(\bar{Y}_{i}|\bar{W}_{i}\right)-h\left(\mathbf{y}_i^n,\bar{Y}_{i}|\bar{W}_{i},W_i\right)\\
&=\sum_{i=1}^{K}
h\left(\mathbf{y}_i^n|\bar{Y}_{i},\bar{W}_{i}\right)-h\left(\mathbf{y}^n_1,\mathbf{y}^n_2,\ldots,\mathbf{y}^n_K|W_1,W_2,\ldots,W_K\right)\\
&+\sum_{i=2}^{K}h\left(\bar{Y}_{i}|\bar{W}_{i}\right)-h\left(\mathbf{y}_i^n,\bar{Y}_{i}|\bar{W}_{i},W_i\right)+h\left(\mathbf{y}^n_2,\ldots,\mathbf{y}^n_K|W_2,\ldots,W_K\right)\\
&\stackrel{(c)}{=} \sum_{i=1}^{K} h\left(\mathbf{y}_i^n|\bar{Y}_{i},\bar{W}_{i}\right)-h\left(\mathbf{y}^n_1,\mathbf{y}^n_2,\ldots,\mathbf{y}^n_K|W_1,W_2,\ldots,W_K\right)\\
&\stackrel{(d)}{=} \sum_{i=1}^{K} h\left(\mathbf{y}_{i}^n\bigg\vert\bar{Y}_{i},\bar{W}_{i},\bar{X}_{i}\right)-\sum_{i=1}^{K}\sum_{t=1}^{n}h\left(\mathbf{z}_i(t)\right)\\
&= \sum_{i=1}^{K} h\left(\sum_{j=1}^{i}\mathbf{H}_{i,j}\mathbf{x}_j^n+\mathbf{z}^n_i\bigg\vert\bar{Y}_{i},\bar{W}_{i},\bar{X}_{i}\right)-\sum_{i=1}^{K}\sum_{t=1}^{n}h\left(\mathbf{z}_i(t)\right)\\
&\stackrel{(e)}{\leq} \sum_{i=1}^{K}
\sum_{t=1}^{n}h\left(\sum_{j=1}^{i}\mathbf{H}_{i,j}\mathbf{x}_j(t)+\mathbf{z}_i(t)\right)-h\left(\mathbf{z}_i(t)\right)
\end{align*}
where $(a)$ follows from the fact that adding information
increases mutual information (providing a genie); $(b)$ follows
from the independence of $(W_1,W_2,\ldots,W_K)$; $(c)$ follows
from the recursive properties of $\bar{W}_i$ and $\bar{Y}_i$;
$(d)$ follows from the fact that $\mathbf{x}_i(t)$ is a function
of $(W_i,\mathbf{y}_{i}^{t-1})$ and $\mathbf{x}_i^n$ is a function
of $(W_i,\mathbf{y}_{i}^{n})$; and $(e)$ follows from the fact
that conditioning reduces entropy.

Therefore, we have
\begin{align*}
\sum_{i=1}^{K}R_i&\leq
\sum^{K}_{i=1}h\left(\sum^{i}_{j=1}\mathbf{H}_{i,j}\mathbf{x}_j+\mathbf{z}_i\right)-h\left(\mathbf{z}_i\right)\\
&\stackrel{(a)}{\leq}\sum^{K}_{i=1}\left(\left(D_d+D_c(i-1)\right)\log
P+o(\log
P)\right)-\sum_{i=1}^{K}h\left(\mathbf{z}_i\right)\\
&=\left(D_dK+\frac{K(K-1)D_c}{2}\right)\log P+o(\log P)
\end{align*}
where $(a)$ follows from the fact that the pre-log term of
$h(\sum^{i}_{j=1}\mathbf{H}_{i,j}\mathbf{x}_j+\mathbf{z}_i)$ is
constrained by $D_d+D_c(i-1)$. Hence, we get the following upper
bound:
\begin{align*}
\Gamma_{fb}\leq \left(D_dK+\frac{K(K-1)D_c}{2}\right).
\end{align*}

\section{Conclusion}
In this paper, we have investigated the total DoF of the $K$-user
rank-deficient interference channel with feedback. When $K=2$, we have developed an explicit achievable scheme and obtaining a matching upper bound, thus characterizing the total DoF.
When $K=3$, we have proposed a new achievable scheme which involves interference alignment, especially when the rank of cross links is sufficiently large as compared to the number of antennas at each node. In addition, we have derived an upper bound
for the general $K$-user case.

We have showed that in contrast to the full-rank case, feedback can
indeed increase the DoF by providing alternative signal paths. Furthermore, if we can use sufficiently many antennas at each node, this DoF gain increases proportionally with the number of users.
 Therefore, using feedback can be an attractive solution to overcome the rank-deficiency of channel matrices in a poor scattering environment.

Our work can be extended to several interesting directions: (1)
Developing an achievable scheme for the general $K$-user case; (2) Extending to other
feedback models (e.g., limited feedback); (3) Extending to the
cases in which there is no or delayed channel state information at
transmitters.

\bibliographystyle{IEEEtran}
\bibliography{IEEEabrv,References}

\begin{thebibliography}{10}
\providecommand{\url}[1]{#1}
\csname url@samestyle\endcsname
\providecommand{\newblock}{\relax}
\providecommand{\bibinfo}[2]{#2}
\providecommand{\BIBentrySTDinterwordspacing}{\spaceskip=0pt\relax}
\providecommand{\BIBentryALTinterwordstretchfactor}{4}
\providecommand{\BIBentryALTinterwordspacing}{\spaceskip=\fontdimen2\font plus
\BIBentryALTinterwordstretchfactor\fontdimen3\font minus
  \fontdimen4\font\relax}
\providecommand{\BIBforeignlanguage}[2]{{%
\expandafter\ifx\csname l@#1\endcsname\relax
\typeout{** WARNING: IEEEtran.bst: No hyphenation pattern has been}%
\typeout{** loaded for the language `#1'. Using the pattern for}%
\typeout{** the default language instead.}%
\else
\language=\csname l@#1\endcsname
\fi
#2}}
\providecommand{\BIBdecl}{\relax}
\BIBdecl

\bibitem{Shannon56}
C.~E. Shannon, ``The zero error capacity of a noisy channel,'' \emph{IRE
  Transactions on Information Theory}, Sept. 1956.

\bibitem{Ozarow84}
L.~H. Ozarow, ``The capacity of the white {G}aussian multiple access channel
  with feedback,'' \emph{{IEEE} Trans. Inf. Theory}, vol.~30, no.~4, pp.
  623--629, July 1984.

\bibitem{LNIT}
A.~{El Gamal} and Y.-H. Kim, \emph{Lecture Notes on Network Information
  Theory}, 2010.

\bibitem{Suh11}
C.~Suh and D.~Tse, ``Feedback capacity of the {G}aussian interference channel
  to within 2 bits,'' \emph{{IEEE} Trans. Inf. Theory}, vol.~57, pp.
  2667--2685, May 2011.

\bibitem{Huang09}
C.~Huang and S.~A. Jafar, ``Degrees of freedom of the {MIMO} interference
  channel with cooperation and cognition,'' \emph{{IEEE} Trans. Inf. Theory},
  vol.~55, no.~9, pp. 4211--4220, Sept. 2009.

\bibitem{Vaze11}
C.~S. Vaze and M.~K. Varanasi, ``The degrees of freedom region of the {MIMO}
  interference channel with {S}hannon feedback,'' [Online]. Available:
  http://arxiv.org/abs/1109.5779, Oct. 2011.

\bibitem{Hwang12}
I.-H. Wang and C.~Suh, ``Feedback increases the degrees of freedom of two
  unicast {G}aussian networks,'' in \emph{Proc. 50th Annu. Allerton Conf.
  Communication, Control, and Computing}, Monticello, IL, Oct. 2012.

\bibitem{Winner}
{D. Baum et al.}, ``{WINNER Deliverable 5.4 v.1.0, Final report on link level
  and system level channel models},'' IST-WINNER, Tech. Rep. IST-2003-507581,
  Nov. 2005.

\bibitem{802.16m}
{\text{IEEE}}.~802.16m 08/004r2, ``{IEEE} 802.16m evaluation methodology
  document {(EMD)},'' IEEE, Tech. Rep., July 2008.

\bibitem{mit}
{J. Bicket, D. Aguayo, S. Biswas, and R. Morris}, ``Architecture and evaluation
  of an unplanned 802.11b mesh network,'' in \emph{Proc. 11th Annual
  International Conference on Mobile Computing, Networking, and Communications,
  2005 (MobiCom '05)}, Cologne, Germany, Aug. 2005, pp. 31--42.

\bibitem{rooftop}
{N. D. Skentos, A. G. Kanatas, and P. Constantinou}, ``{MIMO} channel
  characterization results from short range rooftop to rooftop wideband
  measurements,'' in \emph{Proc. 11th Annual International Conference on Mobile
  Computing, Networking, and Communications, 2005 (MobiCom '05)}, Cologne,
  Germany, Aug. 2005, pp. 137--144.

\bibitem{Chae11}
S.~H. Chae and S.-Y. Chung, ``On the degrees of freedom of rank deficient
  interference channels,'' in \emph{Proc. IEEE International Symposium on
  Information Theory}, Saint Petersburg, Russia, Jul.-Aug. 2011, pp.
  1367--1371.

\bibitem{Krishnamurthy12}
S.~R. Krishnamurthy and S.~A. Jafar, ``Degrees of freedom of 2-user and 3-user
  rank-deficient {MIMO} interference channels,'' in \emph{Proc. 2012 IEEE
  Global Telecommunications Conference (GLOBECOM)}, Anaheim, USA, Dec. 2012.

\bibitem{Zeng12}
Y.~Zeng, X.~Xu, Y.~L. Guan, and E.~Gunawan, ``On the achievable degrees of
  freedom for the 3-user rank-deficient {MIMO} interference channel,''
  [Online]. Available: http://arxiv.org/abs/1211.4198, Nov. 2012.

\bibitem{Tse_wireless}
D.~Tse and P.~Viswanath, \emph{Fundamentals of Wireless communication}.\hskip
  1em plus 0.5em minus 0.4em\relax Cambridge University Press, 2005.

\bibitem{Chae12}
S.~H. Chae, S.-W. Jeon, and S.-Y. Chung, ``Cooperative relaying for the
  rank-deficient {MIMO} relay interference channel,'' \emph{{IEEE} Commun.
  Lett.}, vol.~16, no.~1, pp. 9--11, Jan. 2012.

\bibitem{Papailiopoulos12}
D.~S. Papailiopoulos, C.~Suh, and A.~G. Dimakis, ``Feedback in the ${K}$-user
  interference channel,'' in \emph{Proc. IEEE International Symposium on
  Information Theory}, Cambridge, MA, Jul. 2012, pp. 3130--3134.

\bibitem{Cadambe107}
V.~R. Cadambe and S.~A. Jafar, ``Interference alignment and degrees of freedom
  for the {$K$}-user interference channel,'' \emph{{IEEE} Trans. Inf. Theory},
  vol.~54, no.~8, pp. 3425--3441, Aug. 2008.

\end{thebibliography}

\end{document}